\documentclass[final,12pt]{article}
\usepackage{amssymb}
\usepackage{amsfonts}
\usepackage{amsmath}
\usepackage{geometry}
\usepackage{amsthm}
\usepackage[comma]{natbib}
\usepackage{xcolor}
\usepackage{graphicx}
\usepackage{setspace}
\usepackage[bottom]{footmisc}
\usepackage{hyperref}
\usepackage{scalefnt}
\usepackage{enumitem}
\usepackage{cleveref}
\usepackage{dsfont}

\newtheorem{lemma}{Lemma}

\newtheorem{proposition}{Proposition}

\begin{document}
\begin{spacing}{1.5}
\title{Estimating Social Network Models with Link Misclassification\footnote{               
    We are grateful to seminar and conference participates at CalTech, Chinese University of Hong Kong, London School of Economics, National University of Singapore, Northwestern, Oxford, Texas Camp Econometrics, U Amsterdam, U Chicago, University College London, U Penn, UT Austin, U Warwick, U Wisconsin, Vanderbilt, Wuhan U and Xiamen U for useful comments and suggestions. Lewbel and Tang are grateful for the financial support from National Science Foundation (Grant SES-1919489). Qu thanks the support from the National Natural Science Foundation of China (Project no. 72222007 and 72031006). Any or all remaining errors are our own. }
}

\author{Arthur Lewbel, Xi Qu, and Xun Tang}


\date{ \today }

\maketitle

\begin{spacing}{1}
\begin{abstract}


\noindent We propose a new adjusted 2SLS estimator for social network models when the binary links reported are misclassified (some zeros reported as ones and vice versa). This can happen e.g., when links are imperfectly imputed from observed data. Such misclassification adds new sources of correlation between the regressors and errors, which invalidates conventional estimators. Our novel adjusted 2SLS estimator does not require structural modelling of link formation. Simulation results confirm our method corrects the misclassification bias in standard estimators. We apply our method to study peer effects in household decisions to participate in a microfinance program in Indian villages.

\bigskip 

 \medskip 


\end{abstract}

\end{spacing}

\pagebreak

\section{Introduction}

In many social and economic environments, an individual's behavior or outcome depends not only on his/her own characteristics, but also on the behavior and characteristics of others. 
Call such dependence between two individuals a \textit{link}. A \textit{social network} consists of a group of individuals, some of whom are linked to others. 
The econometrics literature on social networks has largely focused on disentangling various channels of effects based on observed outcomes and characteristics of network members. 
These include identifying the effects on each individual's outcome by (i) the individual's own characteristics (\textit{individual effects}), (ii) the characteristics of people linked to the individual (\textit{contextual effects}), and (iii) the outcomes of people linked to the individual (\textit{peer effects}). 
See \cite{blume2011identification} and \cite{graham2020network} for surveys about identifying such effects in social network models.

A popular approach for estimating social network models is to use two-stage least squares (2SLS). 
This requires researchers to construct instruments for the endogenous peer outcomes, using \textit{perfect knowledge} of the network structure, as given by the \textit{adjacency} matrix (i.e., the matrix that lists all links in the network). 
See, for example, \cite{bramoulle2009identification}, 
\cite{kelejian1998generalized}, 
\cite{lee2007identification}, 
and \cite{lin2010identifying}. 
However, in practice, the network links reported in a sample could be subject to misclassification errors. 

Specifically, the links in a sample are often imputed using survey responses, or some observable metric of ``connectedness''. 
Such reported links practically serve as binary proxies that are correlated with, but \emph{not identical to}, the actual links in the structural form of interest. 
(For example, in a structural form for households' micro-finance decisions, the links between households should reflect the actual exchange or sharing of relevant information. 
Therefore, links imputed by known incidences of their mutual visits could serve as intuitive proxies, but not a perfect measure.)
Other sources for link misclassification include the recall errors by survey respondents, or lapses in data input. 
These misclassification errors are generally \textit{two-sided}: an existing link between two individuals may be misclassified as non-existent, or the sample may erroneously record links between those who are not linked.

Misclassification of links in the sample poses major methodological challenges for estimators like 2SLS. 
To see this, consider a data-generating process (DGP) from which a large number of independent networks (i.e., groups) are drawn. Each group $s$ consists of $n_s$ members\footnote{
    We also consider the case with a single growing network in the Online Appendix, but our results are easiest to illustrate in the context of many independent groups.} 
and a vector of individual outcomes $y\in \mathbb{R}^{n_s}$ is determined by a structural model: 
\[ y = \lambda Gy+X\beta +\varepsilon, \text{  where } E(\varepsilon|X,G)=0. \] 
In this model, the  $n_s$-by-$n_s$ adjacency matrix $G$ contains dummy variables that describe the group's network: its $(j,k)$-th entry equals one if individual $j$ is linked to member $k$, and zero otherwise. 
(In Section \ref{sec:LIM}, we discuss how to extend our method under an alternative, linear-in-means, a.k.a. ``local average'', specification, where the actual $G$ is row-normalized.) 
Here $X$ is an $n_s$-by-$K$ matrix of exogenous covariates, and $\varepsilon $ is an $n_s$-vector of structural errors.
The random arrays $y$, $G$, $X$, and $\varepsilon $ all vary across the groups in the sample, while the coefficients $\lambda $ and $\beta $ are the same across groups. We drop group subscripts for clarity.\footnote{For simplicity we have for now omitted contextual effects, i.e., a term defined as $GX\gamma $, and any group-level fixed effects. Extensions are in the Online Appendix.}

The regressors in the model are $Gy$ and $X$. While $X$ is exogenous, $Gy$ is correlated with $\varepsilon $. 
The issue of simultaneity arises because any individual's outcome depends on, and is determined simultaneously with, the outcomes of other peers. 
A simple estimator of the peer effect $\lambda$ and individual effects $\beta$ that deals with this simultaneity problem is 2SLS, using $GX$ or $G^{2}X$ as instruments for $Gy$, as in \cite{bramoulle2009identification}.\footnote{
    If the model includes contextual effects $GX\gamma $ in its structural form, then $G^{2}X$ can  be used as instruments for $Gy$; otherwise use of $GX$ as instruments suffices.}

But now suppose that, a researcher does not observe $G$ perfectly. Instead, the researcher observes a noisy measure $H$, which differs from $G$ by randomly misclassifying some links in the DGP. 
The goal now is to estimate $\lambda$ and $\beta$ from a \textquotedblleft \textit{feasible}\textquotedblright\ structural form like: 
\begin{equation} 
\label{eq:feasible_sf}
    y = \lambda Hy+X\beta +u,
\end{equation}
where $u\equiv \lbrack \varepsilon +\lambda (G-H)y]$ is a vector of \textit{composite} errors. 

The misclassified links in $H$ aggravate endogeneity issues in (\ref{eq:feasible_sf}) in three important ways. First, they lead to correlation between $X$ and the error $u$ through $\lambda(G-H)y$, due to the measurement error in the adjacency matrix. This component contains $y$, which by the model is correlated with $X$. This means that the regressors $X$ are no longer exogenous.

Second, these misclassified links cause an additional source of endogeneity in $Hy$. Like $Gy$, the feasible $Hy$ is correlated with the model error $\varepsilon $ due to simultaneity. 
But in addition, $Hy$ is also correlated with $u$ through the measurement error $\lambda(G-H)y$. 

Third, misclassification means that, unlike using $GX$ or $G^{2}X$ as instruments when $G$ is perfectly reported, 2SLS estimates based on the \textit{feasible} instruments $HX$ or $H^{2}X$ would be inconsistent as $HX$ correlates with $\lambda(G-H)y$, leading to a failure of instrument exogeneity.

For all these reasons, conventional 2SLS estimators become inconsistent in the presence of misclassification errors in the links.\footnote{
    While we focus on the 2SLS, same arguments apply to show that maximum likelihood, and the generalized least squares estimators based on (\ref{eq:feasible_sf}) are also inconsistent with link misclassification errors. } In this paper, we introduce an \textit{adjusted-2SLS} estimator, which resolves these endogeneity issues and consistently estimates $(\lambda,\beta)$ using alternative valid instruments constructed from $H$ despite the misclassification errors in the links. 
We first introduce the main idea for a benchmark case, where an observed $H$ differs from the true $G$ due to random, two-sided misclassification errors at unknown rates $p_0,p_1\in (0,1)$. Here, $p_1$ is the probability that any existing link is missing in the sample, while $p_0$ is the probability that a non-existent link is erroneously recorded as existing in the sample.\footnote{In the Online Appendix, we extend to allow misclassification rates $p_0(X),p_1(X)$ to depend on covariates.}

Our method is based on a series of new insights that have not been explored in the literature.
\textit{First}, we observe that by adjusting the noisy measure of peer outcomes $Hy$ using the misclassification rates $(p_0,p_1)$, we can restore the exogeneity of $X$ in an \textit{adjusted} feasible structural form. 
Formally, this means if we replace (\ref{eq:feasible_sf}) with: 
\begin{equation}
\label{eq:repara_feasible_sf}
    y = \lambda W_{(H,p_0,p_1)} y +X\beta + v,
\end{equation}
where $W_{(H,p_0,p_1)}$ is a properly designed adjustment of the network measure $H$, then the adjusted composite errors $v \equiv \varepsilon + \lambda[G - W_{(H,p_0,p_1)}]y$ in (\ref{eq:repara_feasible_sf}) satisfy $E(v|X,G) = 0$. 
This holds regardless of how the actual network $G$ is formed, as long as $E(\varepsilon|X,G) = 0$.

\textit{Second}, despite the restored exogeneity of $X$ in (\ref{eq:repara_feasible_sf}), conventional instruments such as $HX$ or $H^{2}X$ remain invalid, because the adjusted errors $v$ still depend on $H$. 
To resolve this issue, we provide alternative functions of $H$ and $X$ that are valid instruments. 
For example, we show that if $H$ is an unsymmetrized measure of $G$, then under some weak conditions $H^{\prime}X$ is uncorrelated with $v$ (where $H^{\prime}$ is the transpose of $H$), \textit{despite} misclassification errors in $H$. 
This result holds regardless of whether $G$ is symmetric (i.e., with all links being undirected) or asymmetric (i.e., consisting of directed links).
Therefore, we can use $H^{\prime}X$ as valid instruments in an adjusted-2SLS where network measures are adjusted by $W_{(H,p_0,p_1)}$.
To the best of our knowledge, no other paper in the literature has proposed $H^{\prime }X$ as instruments. 

Another scenario, which works regardless of whether the observed or actual adjacency matrices are symmetric or not, is when we observe two noisy measures of the same actual $G$.  
An example is our empirical application, where we observe two different reports of who visits whom. This means we observe two different $H$ matrices with independent misclassification errors. We show 2SLS becomes valid if we use one of these matrices to construct the adjustment term $W_{(H,p_0,p_1)}$ and the other to construct instruments.%

Our \textit{third} contribution is to show that under either scenario above (i.e., when the sample reports either a single unsymmetrized noisy measure $H$, or two independent measures that may or may not be symmetrized), we can provide simple methods to identify and estimate the unknown misclassification rates $(p_0,p_1)$.\footnote{
    The approach we take in this step differs from, and is simpler than, other papers that use multiple measures to deal with misclassification in discrete explanatory variables (e.g. \cite{mahajan2006identification}, \cite{lewbel2007estimation}, and \cite{hu2008identification}).
    This is because, for implementing our adjusted-2SLS, it is only necessary to estimate the rates $(p_0,p_1)$, rather than the distribution of outcomes conditional on the actual $G$.}

Building on these insights, we construct adjusted 2SLS estimators for $(\lambda,\beta)$, and provide their limiting distribution as the number of groups in the sample grows to infinity. 
This estimator essentially applies 2SLS to the adjusted peer outcomes $W_{(H,p_0,p_1)}y$ in (\ref{eq:repara_feasible_sf}), using our new instruments and a closed-form, sample analog estimator for the misclassification rates $(p_0,p_1)$.
The estimator is easy to implement, does not require any numerical searches, and Monte Carlo simulations demonstrate its good performance in finite samples. 

A distinctive feature of our method is that it does not require the researchers to impose any structural model of link formation.
Nor does it require specification of the distribution of the latent $G$ given the observed $H$ and $X$.
We show how intuitive restrictions on link misclassification could provide enough leverage to estimate the social effects through 2SLS. 
This is an important advantage, because in many contexts researchers will not have sufficient prior knowledge to reliably specify a link formation model.

We adapt our method under an alternative specification where the actual network structure $G$ is row-normalized. In addition, in the Online Appendix we generalize the model and our estimator in several directions. 
We show how to include contextual effects (a term defined as $GX\gamma$) as well as group-level fixed effects into the structural form in (\ref{eq:repara_feasible_sf}). 
We also allow the misclassification rates $(p_0,p_1)$ to be heterogeneous and depend on covariates in $X$.

Furthermore, we extend our method to a single large network case, where the asymptotics is to increase the number of individuals on a single network, rather than increasing the number of small groups with fixed sizes. For this extension we examine a setting where the sample is partitioned into \textit{approximate} groups, a.k.a. blocks. Sparse links (with diminishing formation rates) exist \textit{between} blocks, but are not recorded in the sample; links within the blocks can be dense and are randomly misclassified. 

Finally, we apply our method to estimate peer effects in household decisions to participate in a microfinance program in Indian villages, using data from \cite{banerjee2013diffusion}. 
The sample matches the individual surveys to the household surveys, yielding a total of 4,149 households from 43 villages in South India. 
The parameter of interest is the peer effect, which reflects how a household's decision is influenced by the participation of other households to which it is linked. 
Survey information about visits between the households provides two symmetrized noisy measures of undirected links (i.e., two symmetrized $H$ measures). 
We estimate the misclassification rates in each of these two measures using our method, and apply these estimated rates in our adjusted 2SLS to estimate the peer effects. 

We find that participation by another linked household increases a household's own participation rate by around 5.1\%.
This effect is economically significant, compared to the average participation rate of 18.9\% in the sample.
We also find ignoring the issue of link misclassification in the noisy measures and applying  conventional 2SLS results in an upward bias in the estimates of these peer effects.
While such bias can be large as our Monte Carlo simulation reveals, it turns out to be modest in this particular application. 
As such, our analysis in \Cref{sec:application} amounts to a robustness check of the native 2SLS estimates against link misclassification, which would not be viable without the proposed adjusted 2SLS. 
Importantly, our broader insight is that the adjusted 2SLS should be used in any environment where the imputed links serve as an intuitive proxy but not a perfect measure of the true adjacency matrix.

\noindent \textbf{Roadmap.} Section \ref{sec:literature} reviews the literature, and explains our contribution.  
Section \ref{sec:iden} specifies the model, and illustrates the main ideas with independent and identical misclassification rates.   
Section \ref{sec:estimation} defines a closed-form estimator for misclassification rates, and provides an adjusted-2SLS estimator for the social effects. 
Section \ref{sec:simulation} presents Monte Carlo simulation results. 
Section \ref{sec:application} applies our method to analyze peer effects in the microfinance participation in India. 
Proofs of the main results are collected in the Appendix. 
Extensions to models with contextual effects, heterogeneous misclassification rates, group fixed effects, as well as the setting of one single large network, are in the Online Appendix. 

\section{Related Literature} \label{sec:literature}

Models with misclassified binary or discrete variables have been studied extensively in the econometrics literature. 
\cite{aigner1973regression}, \cite{klepper1988bounding}, \cite{bollinger1996bounding}, and \cite{molinari2008partial} point-identify or set-identify such models using various restrictions on the misclassification rates; 
\cite{mahajan2006identification}, \cite{lewbel2007estimation}, and \cite{hu2008identification} exploit exogenous instruments to deal with misclassified explanatory variables. 

Our method in this paper differs fundamentally from these predecessors, because in social network models the misclassification in links compounds with the simultaneity in $y$, and renders all covariates endogenous. 
Our solution is new in that it estimates the misclassification rates first, and use them to adjust the network measure and resolve the compounded endogeneity issues. 
In contrast with the papers above, we do not rely on any eigenvalue decomposition algorithm, or external instrument-like variables, to deliver identification. Instead, our identification strategy leads to a closed form estimator using internal instruments constructed from the covariates and the noisy network measures.

Estimation of peer effects in social networks with measurement errors in the links is an increasingly important topic.
\cite{butts2003network} proposes a hierarchical Bayesian model to infer social structure in the presence of measurement errors. 
\cite{shalizi2013consistency} note the challenge of dealing with missing network links in Random Graph Models.
\cite{advani2018credibly} show that even a relatively low misreporting rate can lead to large bias in causal effect estimates.
\cite{chandrasekhar2011econometrics} show how egocentrically sampled network data can be used to predict the full network in a graphical reconstruction process. 
\cite{liu2013estimation} shows that when the adjacency matrix is not row-normalized, instrumental variable estimators based on an out-degree distribution can be valid.
\cite{Comola2014} proposes a test for whether survey-reported social network links can be reliably interpreted as actual connections and for determining whether such links are formed through unilateral (one-sided) or bilateral (mutual) processes.

\cite{hardy2019estimating} estimate treatment effects on a social network when the reported links are a noisy representation of true spillover pathways. 
They use a mixture model that accounts for missing links as unobserved network heterogeneity, and estimate it using an Expectation-Maximization algorithm. 
This approach requires a parametric model of how links are determined and treatment is assigned, and requires enumerating the likelihood conditional on all possible treatment exposures (which in turn depends on the latent unobserved network). 
\cite{auerbach2022identification} studies a network model where links are correctly measured but both peer and contextual effects interact with unobserved individual heterogeneity that affects link formation.
In contrast with these papers, we focus on estimating social effects in linear social network models while fully exploiting implications of randomly misclassified links. 
Our method does not require modeling the formation of actual links; our estimator is an adjusted 2SLS, which has a closed form and is easy to compute.

\cite{liu2013estimation} estimates a social network model when the data consists of a \textit{subset} of individuals sampled randomly from a larger group in the population. 
In his setting, the links and outcomes among this sampled subset of group members are perfectly measured while those of all others are not reported in the data. 
In comparison, we do not study the inference of sampled networks; instead, we let the group memberships be fixed and known, and allow every individual in the sample to have randomly missing links. 
As noted above, this imperfect measure of links leads to the failure of conventional 2SLS in our setting.

\cite{boucher2020estimating} estimate peer effects when the social networks in the sample are subject to measurement issues, such as missing or misclassified links. Their method requires researchers know, or have a consistent estimator of, the distribution of the actual network. 
They construct instruments by drawing from this distribution, and use 2SLS to estimate the peer effects. 
In comparison, the method we propose does not require such prior knowledge or estimates of network distribution.

\cite{griffith2022name} studies the case where links are censored in the sample, and characterizes the bias in a reduced-form regression (i.e., when the outcomes in $y$ are regressed on exogenous covariates $X$ and $GX$). 
For a model with $\lambda=0$, \cite{griffith2022name} shows the bias can be consistently estimated under an order invariance condition, i.e., the covariance of characteristics of those linked to an individual is invariant to the order in which those links are reported or censored.\footnote{
    This condition mitigates the issue of endogenous selection of uncensored links, and in this sense plays a similar role to our assumption of randomly misclassified links.}
\cite{Griffith2023impact} extend this investigation to include both linear-in-sums (where $G$ has binary entries) and linear-in-means (where $G$ is row-normalized). They show how nonzero, structural peer effects $\lambda$ enter the estimand of the reduced-form regression above, as well as how general misclassification, e.g., due to randomly missing links or censored links, affect these estimands.
In comparison, we focus on empirical settings where links are misclassified at random. (This is later generalized to the case with heterogeneous misclassification rates.) 
We show that conventional 2SLS estimands in this case contain bias in peer effects.
Bias correction in our case is immediate once the misclassification rates are estimated using a simple approach that we provide.

\cite{lewbel2023ignoring} show that if the order of measurement errors in links is sufficiently small (e.g., the number of misclassified links in a single, large network does not grow too fast with the sample size), conventional instrumental variables estimators that ignore these measurement errors remain consistent, and standard asymptotic inference methods remain valid. In contrast, in this paper we deal with new challenges outside the scope of \cite{lewbel2023ignoring}. Namely, we allow the misclassification rates to be non-diminishing (fixed) in an asymptotic framework with many independent, finite-sized groups. In such settings, the measurement errors are large enough to invalidate conventional 2SLS estimators.


\section{Model and Identification} \label{sec:iden}

Consider a DGP from which a large number of small, independent networks (groups) are drawn. Each group $s$ consists of $n_s \geq 3$ individuals, with $n_s$ being finite integers.
We first identify and estimate a social network model when links are randomly misclassified in the sample. (In the Online Appendix, we allow misclassification to depend on covariates and we extend to a single large network with sparse links between groups.) 
We establish asymptotic properties of our estimator as the number of groups approaches infinity. 

The structural form for the vector of individual outcomes $y_s\in \mathbb{R}^{n_s}$ in group $s$ is:%
\begin{equation}
y_s=\lambda G_sy+X_s\beta + \varepsilon_s \text{,}  \label{StrucForm}
\end{equation}%
where $\lambda $ and $\beta $ are constant parameters, $X_s$ is an $n_s$-by-$K$ matrix of explanatory variables, and $G_s\in \{0,1\}^{n_s\times n_s}$ is the network adjacency matrix, with its $(i,j)$-th entry $G_{s,ij}=1$ if an individual $i$ is linked to $j$ in group $s$, and $G_{s,ij}=0$ otherwise. The matrix $G_s$ may be asymmetric with directed links ($G_{s,ij} \neq G_{s,ji}$ for some $i \neq j$), or symmetric with undirected links ($G_{s,ij}=G_{s,ji}$ almost surely).
Section \ref{sec:LIM} adapts our method when $G$ is row-normalized.

Let $I_s$ by an $n_s$-by-$n_s$ identity matrix, and assume $(I_s-\lambda G_s)$ is invertible almost surely. 
A sufficient condition for this is $||\lambda G_s||<1$ for \textit{any} matrix norm $||\cdot||$ almost surely.
Solving equation (\ref{StrucForm}) for $y_s$ gives the reduced form for outcomes: 
\begin{equation}
y_s=M_s(X_s\beta + \varepsilon_s )\text{, where }M_s\equiv (I_s-\lambda G_s)^{-1}\text{. }
\label{ReduForm}
\end{equation}

We do not observe the actual network $G_s$. Instead, the sample reports a noisy measure $H_s\in\{0,1\}^{n_s\times n_s}$. 
That is, for unknown pairs of individuals $i \neq j$, $G_{s,ij}$ is randomly misclassified as $H_{s,ij}=1-G_{s,ij}$.
By convention, let $G_{s,ii}=0$ and $H_{s,ii}=0$ for all $i$ and $s$.

For simplicity, let the group sizes $n_s = n $ be fixed across groups for now. We will add back the group subscripts and allow group size variation in Section \ref{sec:estimation}. 
\subsection{Assumptions}
\label{sec:assum_H}
We maintain the following conditions on the noisy measure $H$ throughout Section \ref{sec:iden}:
\begin{eqnarray*}
&&\text{(A1) }E(H_{ij}|G,X)=E(H_{ij}|G_{ij},X)\text{ for all $i$ and $j$;}
\\
&&\text{(A2) }E(H_{ij}|G_{ij}=1,X)=1-p_1\text{, }E(H_{ij}|G_{ij}=0,X)=p_0\text{, and }p_0+p_1<1\text{ for all  $i$}  \not= \text{  $j$;}
\\
&&\text{(A3) }E(\varepsilon |G,X,H)=0\text{.}   
\end{eqnarray*}
Condition (A1) states the incidence of misclassifying a link is conditionally independent from the actual status of all other links.
This condition doesn't allow situations where the chance of misreporting a link depends on other links.

Under (A2), misclassification probabilities  conditional on actual link status are fixed at $p_0$ and $p_1$ respectively, and are independent from $X$. 
With $\Pr\{G_{ij}=1\}<1$, the inequality constraint ``$p_0+p_1<1$'' is equivalent to ``$H_{ij}$ and $G_{ij}$ are positively correlated.'' 
That is, the noisy measure is positively correlated with the actual link status despite the misclassification error. 
This is standard in the literature on misclassified regressors, e.g., \cite{bollinger1996bounding}, \cite{hausman1998misclassification}.
In the Online Appendix (Section 1.2), we extend our identification strategy to allow heterogeneous misclassification rates that vary with the covariates $X$.

Condition (A3) rules out endogeneity in link formation, assuming $(G,X,H)$ are exogenous to structural errors $\varepsilon $.

Conditions (A1) and (A2) hold jointly in two common scenarios. 
In the first scenario, which we refer to as \textit{unsymmetrized} measures, each $(i,j)$-th entry in $H$ is an \textit{independent} measure of $G_{ij}$.
For example, $H_{ij}$ (or $H_{ji}$) reports individual $i$'s (or $j$'s) binary response to a survey question about whether a link exists between $i$ and $j$. 
A measure $H$ constructed this way is flexible in that it allows the researcher to remain agnostic about whether the actual $G$ is symmetric with undirected links or not.
This is also an intuitive way to construct $H$ when the actual $G$ is \textit{known} to be asymmetric with directed links.
In this scenario, if misclassification of $G_{ij}$ happens independently at rates $p_0$ or $p_1$ across links (depending on whether $G_{ij}=1$ or $0$), then (A1) and (A2) are satisfied. 
To reiterate, (A1) and (A2) hold in this first scenario, regardless of whether the actual $G$ is symmetric or not.

In the second scenario, which we refer to as \textit{symmetrized} measures, the actual $G$ is \textit{known} to be symmetric with undirected links, and hence the researcher chooses to symmetrize $H$.
For example, the researcher asks $i$ and $j$ whether they have an undirected link, and records their responses respectively. 
The researcher then constructs a symmetrized measure by setting $H_{ij}$ and $H_{ji}$ both to $1$ if \textit{either $i$ or $j$} responds positively, and both to $0$ otherwise. 
Suppose the responses from $i$ or $j$ independently misclassify an \textit{existing} link at rate $\varphi_1 >0$ (say, due to idiosyncratic recall errors).
Then $\Pr\{H_{ij}=0|G_{ij}=1\} = p_1 \equiv \varphi_1^2$.   
Likewise, if $i$ and $j$ independently misclassify a \textit{non-existent} link at rate $\varphi_0$, then $\Pr\{H_{ij}=1|G_{ij}=0\} = p_0 \equiv 1-(1-\varphi_0)^2$.
Thus, in the second scenario, (A1) and (A2) hold with $\Pr\{H_{ij} = H_{ji}\}=1$ and with the two entries sharing the same misclassification rates $p_1$ and $p_0$ specified above.

On the other hand, (A1) \textit{does} rule out a third, empirically less plausible scenario, in which the actual $G$ is asymmetric with directed links but researchers mistakenly impose a symmetrized $H$ using independent measures of $G_{ij}$ and $G_{ji}$ as in the second scenario. 
In this case, the equality in (A1) fails in general because $E(H_{ij}|G_{ij}=1,G_{ji}=1) = 1-\varphi_1^2 $ while $E(H_{ij}|G_{ij}=1,G_{ji}=0) = \varphi_0 + (1-\varphi_1) - \varphi_0 (1-\varphi_1)$.

A clear advantage of the method we propose is that it allows researchers to consistently estimate social effects while being agnostic about whether the actual links in $G$ are directed or not. 
Our method only requires the noisy measure $H$ satisfy (A1)-(A3), which is not confined to the (a)symmetry of $G$ or $H$.
We recommend a simple guideline for practitioners: if a researcher is unsure about whether the actual links in $G$ are directed or undirected, a safe approach is to construct an unsymmetrized measure $H$ as in the first scenario, and apply our method in this paper to deal with possible misclassification of the links.

It is also important to note that (A1)-(A3) do \textit{not} specify how the actual links in $G$ are formed. Nor do they impose any structure that can be used to derive a conditional likelihood for the actual network, which is $\Pr\{G|H,X\}=\frac{\Pr(H|G,X) \Pr(G|X)}{\sum_{G'}\Pr(H|G',X)\Pr(G'|X)}$.
Constructing such a likelihood requires specifying the DGP of the actual network $\Pr\{G|X\}$, which we refrain from doing. 
Our method therefore differs qualitatively from alternative methods which either use graphical reconstructions such as \cite{chandrasekhar2011econometrics}, or require knowledge of the distribution of actual network matrix such as \cite{boucher2020estimating}.  

Define an \textit{infeasible, adjusted} adjacency matrix: $W\equiv W_{(H,p_0,p_1)}\equiv \frac{H-p_0(\iota\iota'-I)}{1-p_0-p_1}$,  
where $\iota$ is a vector of ones. For the rest of this paper, we suppress the subscripts indicating the arguments $(H,p_0,p_1)$ in $W$ to simplify notation.
That is, $W_{ij} = (H_{ij} - p_0)/(1-p_0-p_1)$ for $i\neq j$, and $W_{ii} = H_{ii} = 0$.
Under (A1) and (A2), $E(W_{ij}|G,X)=1$ whenever $G_{ij}=1$, and $E(W_{ij}|G,X)=0$ whenever $G_{ij}=0$ (including the case with $i=j$).
Thus,%
\begin{equation}
E(W|G,X)=G\text{.}  \label{EQ1}
\end{equation}%
Next, we will exploit this property in (\ref{EQ1}) to establish a useful intermediate result: despite link misclassification, structural parameters $(\lambda,\beta)$ could be consistently estimated by an adjusted 2SLS if the misclassification rates $p_0,p_1$ were known.

\subsection{Infeasible two-stage least squares}
\label{sec:augmentBias}
We write a new \textit{adjusted} structural form using $W$:%
\begin{equation}
y=\lambda Wy + X\beta + \underset{\equiv v}{\underbrace{\varepsilon
+\lambda \left(G- W\right) y}}\text{.}  \label{SF1}
\end{equation}%
This is infeasible as $W$ is a function of the unknown misclassification rates $p_0$ and $p_1$. 
Lemma \ref{lm:exogenous_v} shows $X$ is mean independent with its composite errors $v$, despite link  misclassification. 

\begin{lemma}
\label{lm:exogenous_v} Under (A1), (A2), and (A3), $E(v|G,X)\equiv E[\varepsilon
+\lambda \left(G- W\right) y|G,X]=0$.
\end{lemma}

This lemma is fundamental for our method; it restores exogeneity of $X$ by adjusting the structural form properly to account for link misclassification.  

The importance of Lemma \ref{lm:exogenous_v} is best illustrated in contrast with the \textit{naive} structural form in (\ref{eq:feasible_sf}), i.e., $ y = \lambda Hy + X\beta + u$, which ignores misclassification errors and simply uses the reported $Hy$ as peer outcomes on the right-hand side. 
The composition errors in (\ref{eq:feasible_sf}) are: 
\begin{eqnarray}   
    u   = \varepsilon + \lambda(G-H)y    = v + \left(\frac{p_0+p_1}{1-p_0-p_1}\right)\lambda Hy - \left(\frac{p_0}{1-p_0-p_1}\right)\lambda (\iota\iota'-I)y. \label{eq:relate_compErr}
\end{eqnarray}
While $E(v|G,X)=0$ by Lemma \ref{lm:exogenous_v}, the second and third terms on the right-hand side of (\ref{eq:relate_compErr}) do not satisfy such mean independence. 
Therefore, in a simple, feasible structural form that uses $Hy$ instead of $Wy$, the covariates in $X$ are generally endogenous due to the ignored misclassification errors. 
Later we show such endogeneity leads to an ``\textit{augmentation bias}'' in the 2SLS estimation of (\ref{eq:feasible_sf}) when misclassification is one-sided ($p_0=0$). 

Lemma \ref{lm:exogenous_v} may seem surprising ex ante, because one would expect $(G,X)$ to be correlated with the composite error $v$ which depends on $y$. 
The intuition for the exogeneity is as follows.
Once we condition on the actual network $G$ and $X$, randomness in individual outcomes $y$ is solely due to the actual structural errors $\varepsilon$, which are uncorrelated with both $X$ and $(H,G)$ under (A3). 
As a result, any potential correlation between $v$ and $(G,X)$ could only be due to the measurement error $\lambda (G - W)y$. 
But the property established in (\ref{EQ1}) and the exogeneity of $\varepsilon$ in (A3) imply this measurement error is mean-independent from $(G,X)$. 

Note that we \emph{cannot} use the exogeneity established in \Cref{lm:exogenous_v} alone to construct GMM estimators for $(\lambda,p_0,p_1)$, as it does not suffice for the joint identification of these parameters.
This is seen in the special case of  one-sided misclassification $(p_0=0)$, where the moment condition due to \Cref{lm:exogenous_v} simplifies to $E(y-\frac{\lambda}{1-p_1}Hy | G, X)=X\beta$, which is not sufficient for recovering $\lambda$ and $p_1$ separately \emph{even if} $G$ were perfectly observed in the DGP.

Our goal for the rest of \Cref{sec:iden} is to combine the exogeneity attained in \Cref{lm:exogenous_v} with further information, such as instruments and multiple measures $H$, to identify all model parameters, including the  misclassification rates.
First off, note the term $Wy$ in (\ref{SF1}) remains endogenous, \emph{even if} the misclassification rates were known and used to construct the adjusted measure $W$. This is because $E[\left( Wy\right) ^{\prime }v]\not=0$ in general.\footnote{
    Under (A1) and (A2), $E(W^{\prime}G|G,X)=G^{\prime}G$, but $E(W^{\prime}W|G,X)\neq G^{\prime}G$ in general.
    This is because the $i$-th diagonal entry in $W'W$ is $\sum_k W_{ki}^2$ while its $(i,j)$-th off-diagonal entry is $\sum_k W_{ki}W_{kj}$.
    It then follows from (A3) and the law of iterated expectation that $E\left( y^{\prime }W^{\prime }Wy\right) \not= \lambda E(y^{\prime }W^{\prime}Gy)$ in general.} 

We next consider 2SLS estimation of equation (\ref{SF1}). 
Let $R\equiv (Wy,X)$. Suppose we had a set of instruments $Z$ for $R$. 
By Lemma \ref{lm:exogenous_v}, $Z$ can include $X$, so we only need an additional instrument for $Wy$. 
We will later provide some possible instruments for $Wy$. But for now, just consider what properties any such matrix of instruments $Z$ must satisfy: $Z$ must be an $n$-by-$L$ matrix with $L\geq K+1$ such that $E(Z^{\prime }v)=0$ and the following rank condition holds:%
\begin{equation*}
\text{(IV-R) }   E(Z^{\prime }R)\text{ and }E(Z^{\prime }Z)\text{ have full column
rank.}
\end{equation*}%
Let $\Pi \equiv \left[ E(Z^{\prime }Z)\right] ^{-1}E(Z^{\prime }R)$. 
By (\ref{SF1}) and Lemma \ref{lm:exogenous_v},
\begin{eqnarray}
\Pi ^{\prime }E(Z^{\prime }y) &=&\Pi ^{\prime }E(Z^{\prime }R)(\lambda,\beta ^{\prime })^{\prime }+\Pi ^{\prime }E(Z^{\prime }v)  
\Rightarrow ( \lambda,\beta ^{\prime }) ^{\prime }= 
\left[ \Pi ^{\prime }E(Z^{\prime }R)\right] ^{-1}\left[ \Pi ^{\prime
}E(Z^{\prime }y)\right] \text{.}  \label{EQ2}
\end{eqnarray}%

\begin{proposition} \label{pn:1}
Suppose (A1), (A2), and (A3) hold, and that (IV-R) holds for instruments $Z$. 
The two-stage least-squares estimand using $Z$ for (\ref{SF1}) is $(\lambda,\beta ^{\prime})^{\prime }$.
\end{proposition}

Using $Wy$ instead of $Hy$ as the first regressor in $R$ is crucial for consistency in Proposition \ref{pn:1}. To see why, suppose one applies 2SLS to (\ref{eq:feasible_sf}) using $Hy$ in the regressors $\check R \equiv (Hy,X)$, so the resulting model errors are $u$ as defined in (\ref{eq:relate_compErr}). 
Then the 2SLS estimand would be $
(\lambda,\beta')' + [\check \Pi ^{\prime}E(Z^{\prime} \check R) ] ^{-1} [\check\Pi E(Z'u)]$,
where $\check \Pi$ is similar to $\Pi$, only with $R$ replaced by $\check R$. Endogeneity bias arises in general when $Z$ has components in $X$. This is because $X$ is correlated with the latter two terms on the right-hand side of (\ref{eq:relate_compErr}) through $y$.

In the special case with one-sided misclassification (i.e., $p_0=0$ and $p_1>0$ so that actual links are missing at random, but the sample never reports links that do not exist), we can show $ E(Z'u) = (\frac{p_1}{1-p_1}) E[Z' \check R (\lambda,\mathbf 0)'] $. Consequently, the 2SLS estimand in this case is $(\frac{\lambda}{1-p_1},\beta')'$, indicating an ``\textit{augmentation}'' bias in the peer effect estimator.  

Based on Proposition \ref{pn:1}, we have two main requirements for estimation. First, we need to construct a valid instrument for $Wy$. One possibility is nonlinear functions of $X$, if they correlate with the link formation in $G$ and satisfy the rank condition in (IV-R). However, nonlinear functions of $X$ may be weak instruments as the structural model is linear in $X$.
So, instead we show in Section \ref{sec:use_H} how to construct valid instruments using $H$ and $X$.

The second requirement is that we need to identify and estimate the unknown misclassification rates $p_0$ and $p_1$ to construct $W$. We address this question in Section \ref{sec:get_p}.

\subsection{Constructing instruments from network measures} \label{sec:use_H}

We propose two options for constructing IVs, depending on the number of measures available and whether the measures are symmetrized in the sample. 
The \emph{symmetrization} of a network measures $H$ in a sample is a different concept than the \emph{symmetry} of an actual adjacency $G$. 
The latter is defined by $G$ having undirected links (i.e., $G_{ij}=G_{ji}$ for all $i\neq j$), while the former by the constraint $ H_{ij} = H_{ji} $  for all $i\neq j$ in sample collection.

\subsubsection{Instruments using a single unsymmetrized measure} \label{subsec:use_1H}

Suppose the sample reports a single, unsymmetrized network measure $H$.   
Assume:
\begin{equation*}
\text{(A4) Conditional on }(G,X)\text{, }H_{ij}\text{ and }H_{kl}\text{ are
independent whenever }(i,j)\not=(k,l)\text{.}
\end{equation*}%
This condition states that different links are misclassified independently conditional on the actual link status.
This does \textit{not} restrict whether the actual network $G$ is symmetric or not. 
For example, $H$ may be an \textit{unsymmetrized} measure of $G$ as defined in the first scenario under (A1)-(A2) in Section \ref{sec:assum_H}). 
In this case, (A4) holds when $H_{ij}$ and $H_{ji}$ are independent measures of $G_{ij}$ and $G_{ji}$ respectively, \emph{regardless of} whether $G_{ij}=G_{ji}$ in the actual $G$. 

On the other hand, (A4) fails when $H$ is a \textit{symmetrized} measure, as $H_{ij}$ and $H_{ji}$ are identical and hence cannot be independent.
To deal with this case of symmetrized measures, we give an alternative method for constructing instruments in Section \ref{sec:use_2H}.

\begin{proposition} \label{pn:HX} 
Under (A1)-(A4), $E(Z^{\prime}v)=0$ for $Z\equiv (H^{\prime }X,X)$ or $Z\equiv (W^{\prime }X,X).$
\end{proposition}

Proposition 2 suggests using $H'X$ or $W'X$ as instruments for $Wy$. Recall that $GX$ are valid instruments for $Gy$ if $G$ were perfectly observed. 
Therefore, one may wonder why $H^{\prime }X$ are valid instruments while $HX$ are not. 
To give some intuition, observe that the composite error $v$ in (\ref{SF1}) contains $\lambda (G-W)$ and so includes $H$ through $W$. 
The covariance of this error with $HX$ contains the conditional variance of $H$, which cannot be zero. Therefore, the error $v$ is correlated with $HX$. 
In contrast, the corresponding terms in the covariance of $v$ with $H^{\prime }X$ are conditional covariances of $H_{ij}$ with $H_{ji}$, which by (A4) are zero.
Hence $H'X$ satisfies instrument exogeneity while $HX$ does not. 

In addition to validity, instruments $Z$ need to also satisfy the rank condition (IV-R). The next proposition specifies sufficient conditions for $Z\equiv (W^{\prime }X,X)$ to satisfy (IV-R).
These conditions are primitive in terms of moments of functions of $(X,G)$.%

\begin{proposition}
\label{pn:IVR} 
Under (A1)-(A4), (IV-R) holds for $Z\equiv (W^{\prime }X,X)$ if
\begin{equation}
\left(
\begin{array}{cc}
E(X^{\prime }X) & E(X^{\prime }M^{-1}X) \\
E(X^{\prime }MX) & E(X^{\prime }X)%
\end{array}%
\right) \text{ and }\left(
\begin{array}{cc}
E(X^{\prime }G^{2}X) & E(X^{\prime }GX) \\
E(X^{\prime }GX) & E(X^{\prime }X)%
\end{array}%
\right) \text{ are non-singular.}  \label{SuffCond_IVR}
\end{equation}
\end{proposition}

These primitive conditions serve to rule out ``knife-edge'' cases where the link formation process is aligned with the regressor distribution in such a pathological way that the rank of moments above is reduced. 
Our simulation shows (\ref{SuffCond_IVR}) holds even for restrictive cases where dyadic links are formed as i.i.d. Bernoulli, and independent from $X$. On the other hand, (\ref{SuffCond_IVR}) fails in some other special cases. 
One example is the linear-in-means social interactions model where all members in a group are linked so that $G$ is the product of 1/n and a square matrix of ones and $G^2=G$. Note such a social interactions model would not be identified, due to the ``reflection'' problem as defined in \cite{manski1993identification}. See, e.g., \cite{bramoulle2009identification}, who require that $I$, $G$, and $G^2$ be linearly independent. 

\subsubsection{Instruments using multiple measures} \label{sec:use_2H}

Section \ref{subsec:use_1H} assumes the sample reports a single \textit{unsymmetrized} measure $H$. 
Now we provide an alternative, complementary method for constructing instruments when the sample provides two (or more) $H$, regardless of whether the measures are symmetrized or not. 

For example, \cite{banerjee2013diffusion} provide multiple measures of symmetrized links between households. For each pair of households, the survey asks which households you visited, and which ones visited you. 
\cite{banerjee2013diffusion} symmetrize each of these two measures, yielding symmetric matrices we call $H^{(1)}$ and $H^{(2)}$. 
These two matrices are both measures of the same underlying symmetric network $G$ (where $G_{ij}$ is one if either $i$ visited $j$ or $j$ visited $i$, and zero otherwise).
However, as we show in Section \ref{sec:application}, these two matrices empirically differ substantially, indicating that they are different noisy measures of $G$.

Suppose we observe two matrices, $H^{(1)}$ and $H^{(2)}$, which satisfy (A1), (A2), (A3), and
\begin{equation*}
\text{(A4') Conditional on }(G,X)\text{, }H_{ij}^{(1)}\text{ and }%
H_{kl}^{(2)}\text{ are independent for all }(i,j)\text{ and }(k,l)\text{.}
\end{equation*}%
These two measures $H^{(1)}$ and $H^{(2)}$ have their own misclassification rates, denoted $(p^{(t)}_0,p^{(t)}_1)$ for $t=1,2$ respectively. 
Condition (A4') is plausible when these distinct measures are constructed independently using responses from separate survey questions. Define

\[W^{(t)}\equiv W^{(t)}_{(H,p_0,p_1)}\equiv \frac{H^{(t)}-p^{(t)}_0(\iota\iota'-I)}{1-p^{(t)}_0-p^{(t)}_1}.\]
Using either $W^{(1)}$ or $W^{(2)}$, we can construct a structural form. 
That is, for $t=1,2$,%
\begin{equation}
y=\lambda W^{(t)}y+X\beta +v^{(t)}\text{, where }%
v^{(t)}=\varepsilon +\lambda \left[ G-W^{(t)}\right] y%
\text{.}  \label{eq:2HSF}
\end{equation}%
Under (A1)-(A3) and (A4'), and by an argument similar to Proposition \ref{pn:HX}, we can show that $W^{(2)}X$ and $H^{(2)}X$ satisfy instrument exogeneity with regard to $v^{(1)}$:%
\begin{equation*}
E\left[ (W^{(2)}X)'v^{(1)}\right] = \frac{1}{1-p^{(2)}_0-p^{(2)}_1}E\left[ (H^{(2)}X)'v^{(1)}\right] =0\text{,} \notag
\end{equation*}%
and likewise with $W^{(2)}$ replaced by $H^{(2)}$.
A symmetric result holds by swapping the indexes $t=1,2$. We can then use either $H^{(1)}X$ or $W^{(1)}X$ as instruments for $W^{(2)}y$, or use either $W^{(2)}X$ or $H^{(2)}X$ as instruments for $W^{(1)}y$.

Note that unlike Section \ref{subsec:use_1H} which required an unsymmetrized $H$, the use of multiple $H^{(t)}$ matrices here works regardless of whether each $H^{(t)}$ is symmetrized or not.

\subsection{Recovering misclassification rates \label{sec:get_p}}

To construct $W$ and apply 2SLS, we still need to identify and estimate misclassification rates $p_{0}$ and $p_{1}$. 
Now we show how to leverage variation in $X$ that affects actual link formation and recover these rates from the observation of noisy network measures.

\subsubsection{Using two conditionally independent measures} \label{subsubsec: 2H}

We start with the case of two independent measures $H^{(1)}$ and $H^{(2)}$, which have misclassification rates $(p_{0}^{(t)},p_{1}^{(t)}) $ for $t=1,2$ respectively, and satisfy (A1), (A2), (A3), and (A4').\footnote{It is worth noting that this case is flexible enough to accommodate \textit{both scenarios} in Section \ref{sec:assum_H}. That is, the two independent measures $H^{(1)},H^{(2)}$ may either be \textit{unsymmetrized} or \textit{symmetrized}, as introduced in Section \ref{sec:assum_H}. Recall that in the first scenario researchers do not know whether the actual adjacency $G$ is symmetric or not, while in the second scenario researchers do know the actual $G$ is symmetric.} 

 Assume $X$ correlates with network formation. Specifically, assume we can define a function $\phi_{ij}(X)$ that is related to the distribution of $G_{ij}$. 
 In the simplest case, $\phi_{ij}(X)$ would be binary-valued, with $\Pr\{G_{ij}=1\vert \phi_{ij}(X)=0\} \neq \Pr\{G_{ij}=1\vert \phi_{ij}(X)=1\}$. Note this imposes no restriction on the true link formation other than being correlated in some way with $X$. We can accommodate polar extreme cases, such as endogenous network formation based on pairwise stability, where $G_{ij}$ depends on the demographics of all group members' $X$, versus dyadic link formation models where $G_{ij}$ depends only on pair-specific demographics $(X_i,X_j)$.        

To illustrate, in our empirical study in Section \ref{sec:application}, we define $\phi_{ij}(X)=1$ if $i$ and $j$ are from the same caste, i.e., $\phi_{ij}(X) = \mathds{1}\{X_{i,1} = X_{j,1}\}$ where $X_{i,1}$ denotes $i$'s caste and $\mathds{1}\{\cdot\}$ is the indicator function. The working assumption is that two households of the same caste have a different probability of forming a link than those from different castes. 

The intuition for our identification is as follows. Let $\pi_{1}$ be the unknown average probability that a cell $G_{ij}=1$, conditional on $\phi_{ij}(X)=1$. 
Then consider $\Pr\{H_{ij}^{(t)}=1\vert \phi_{ij}(X)=1\}$, which is a known function of $\pi_{1}$, $p_{0}^{(t)}$, and $p_{1}^{(t)}$ for $t=1,2$. 
This provides two equations (one for each value of $t$) in the unknown misclassification rates and in $\pi_{1}$. The same construction conditioning on $\phi_{ij}(X)=0$ gives two more equations in the unknown misclassification rates and in $\pi_{0}$. 
Finally, the conditional average probability of the product $H_{ij}^{(1)}H_{ij}^{(2)}=1$ gives two more equations for identification.

Making this logic precise, define $\pi _{1}\equiv \frac{1}{n(n-1)}\sum_{i\neq j} \Pr \{G_{ij}=1|\phi_{ij}(X)=1\}$. Consider the following set of three conditional moments of $H_{ij}^{(1)}$ and $H_{ij}^{(2)}$:
%
\begin{eqnarray}
\frac{1}{n(n-1)}\sum_{i\neq j}E\left[ \left. H_{ij}^{(1)}H_{ij}^{(2)}\right\vert \phi_{ij}(X)=1 \right] &=&\left(
1-p_{1}^{(1)}\right) \left( 1-p_{1}^{(2)}\right) \pi_{1}+p_{0}^{(1)}p_{0}^{(2)}(1-\pi _{1})\text{;}  \nonumber \\
\frac{1}{n(n-1)}\sum_{i\neq j} E\left[ \left. H_{ij}^{(t)}\right\vert \phi_{ij}(X)=1 \right] &=&\left(
1-p_{1}^{(t)}\right) \pi _{1}+p_{0}^{(t)}(1-\pi _{1})\text{ for }t=1,2\text{.%
}  \label{eq:id_misc_rates}
\end{eqnarray}%
Note these are three distinct equations as the second applies for $t=1$ and $t=2$. We obtain three more equations by replacing $\phi_{ij}(X)=1$ with $\phi_{ij}(X)=0$ and replacing $\pi _{1}$ with $\pi _{0}$. 

The left-hand side of these six equations can be estimated from observed $H^{(1)}$, $H^{(2)}$, and $X$, while the right-hand sides are functions of six unknown parameters: $\pi _{1},\pi _{0}$ and $p_{1}^{(t)},p_{0}^{(t)}$ for $t=1,2$. 
Assume that $\pi _{1} \neq \pi _{0}$, meaning that $\phi_{ij}(X)$ does affect the true link formation. 

Despite the nonlinearity of these six equations, we show that they can be uniquely solved for the six unknown parameters, and in particular we provide closed-form expressions for the misclassification rates $p_{1}^{(t)},p_{0}^{(t)}$ for $t=1,2$. See the proof in Appendix A2 for details.

This identification requires choosing a function $\phi_{ij}(\cdot)$ such that the probability of link formation is different for the event $\{\phi_{ij}(X)=1\}$ than when $\{\phi_{ij}(X)=0\}$ so $\pi_1 \neq \pi_0$. 

It is worth emphasizing that we can generalize this identification argument to broader settings with other choices of $\phi_{ij}(\cdot)$, including those with a continuous range. %
For instance, $\phi_{ij}(X)$ may return a continuous measure of difference between the demographic features of $i$ and $j$. 
In this case, we partition the range of $\phi_{ij}(X)$ into mutually exclusive subsets, denoted by $\boldsymbol{\phi^0}$ and $\boldsymbol{\phi^1}$. Then define $\pi_\varsigma\equiv \frac{1}{n(n-1)}\sum_{i\neq j} \Pr \{G_{ij}=1|\phi_{ij}(X)\in\boldsymbol{\phi^\varsigma}\}$ for $\varsigma = 0,1$. 
The identification strategy applies with 
$\{\phi_{ij}(X)=\varsigma\}$ replaced by $\{\phi_{ij}(X) \in \boldsymbol{\phi^\varsigma}\}$ respectively.

Our focus here is on recovering misclassification rates.
We treat $\pi_1,\pi_0$ as ``nuisance'' parameters that are identified as a by-product in our identification of $p^{(t)}_1, p^{(t)}_0$ for $t=1,2$. 
We do not exploit knowledge of $\pi_1,\pi_0$ for estimation, or to infer the link formation process. 


\subsubsection{Using a single, unsymmetrized measure} \label{subsubsec: 1H}

The identification method of the previous section can be readily modified to recover the misclassification probabilities in the case with a single, unsymmetrized measure $H$ when the actual $G$ is \textit{known} to be symmetric with undirected links. 
Suppose $H$ satisfies (A1)-(A4) with misclassification rates $p_{1},p_{0}$. 
For any \textit{unordered} pair $i\not=j$, construct two noisy measures $H_{\{i,j\}}^{(1)}\equiv H_{ij}$ and $H_{\{i,j\}}^{(2)}\equiv H_{ji}$. 
(The new subscripts for $H^{(t)}$, i.e., $\{i,j\}$, only serve as a reminder that these two measures are symmetrized.)
We then obtain a system of equations similar to (\ref{eq:id_misc_rates}), only with $\frac{1}{n(n-1)}$, $\sum_{i\neq j} $, $ H_{ij}^{(t)} $, $\phi_{ij}$ replaced by $\frac{2}{n(n-1)}$, $\sum_{i>j}$, $ H_{\{i,j\}}^{(t)} $, $\phi_{\{i,j\}}$ respectively, and with identical rates across the measures, i.e. $p_{1}^{(t)}=p_{1}$ and $p_{0}^{(t)}=p_{0}$ for $t=1,2$. 
The same argument then identifies $\pi _{1},\pi_{0},p_{1},p_{0}$ using variation in $\phi_{\{i,j\}}(X)$.

\subsection{Concluding remarks about identification}

\noindent \textbullet \quad \textbf{(A panorama view of all scenarios covered.)}
The methods proposed in Section \ref{sec:iden} are flexible enough to accommodate various scenarios defined by whether the actual adjacency matrix $G$ is symmetric or not, and whether the observed network measure(s) $H$ is symmetrized or not. 
The table below summarizes the solutions of adjusted 2SLS that we propose for each one of those scenarios. \bigskip 

\begin{tabular}{ccccccc}
& \multicolumn{6}{c}{Reported Network Measures $H$} \\ \cline{2-7}
& \multicolumn{2}{|c}{Single, unsym'zed} & \multicolumn{2}{|c}{Multiple, sym'zed} & \multicolumn{2}{|c|}{Multiple, unsym'zed} \\ \cline{2-7}
& \multicolumn{1}{|c}{(IV)} & (MR) & \multicolumn{1}{|c}{(IV)} & (MR) & 
\multicolumn{1}{|c}{(IV)} & \multicolumn{1}{c|}{(MR)} \\ \cline{2-7}
Sym. $G$ & \multicolumn{1}{|c}{ $\ \text{Sec 3.3.1 }$} & Sec 3.4.2 & 
\multicolumn{1}{|c}{ \ Sec 3.3.2 \ } & Sec 3.4.1 & \multicolumn{1}{|c}{ \ \ Sec 3.3 \ \ }
& \multicolumn{1}{c|}{Sec 3.4} \\ \cline{2-7}
Asym. $G$ & \multicolumn{1}{|c}{ Sec 3.3.1 } & see text & \multicolumn{2}{|c}{
violates (A1)} & \multicolumn{1}{|c}{Sec 3.3.2} & \multicolumn{1}{c|}{Sec 3.4.1} \\ 
\cline{2-7}
\end{tabular}

\bigskip \bigskip

\noindent Each one of the six cells in last two rows represents a particular scenario, defined by the (a)symmetry of the actual adjacency $G$ \textit{as well as} the number and property of network measures $H$ available. 
Solutions for estimating $\lambda$ and $\beta$ in each scenario consist of two parts: construction of instruments (IV), and recovery of misclassification rates (MR). 

If the actual $G$ is symmetric and the sample reports a single, unsymmetrized $H$, one can recover MRs using Section 3.4.2 and construct IVs using Section 3.3.1. 
Likewise, if the actual $G$ is asymmetric and the sample reports multiple, unsymmetrized $H$, one can recover MRs using Section 3.4.1 and construct IVs using Section 3.3.2. 
If the actual $G$ is symmetric and the sample report multiple, unsymmetrized $H$, then one can recover MRs using \textit{either} approach in Section 3.4, and construct IVs using \textit{either} approach in Section 3.3.

For the scenario with an asymmetric $G$ and a single, unsymmetrized measure, our paper presents a valid way to construct instruments, but does not propose a way to identify the MRs. 
To perform the latter task, one may adopt a method from \cite{hausman1998misclassification} to a dyadic link formation model. 
We do not elaborate on that method as it requires a link formation model, which we have intentionally refrained from doing throughout this paper. 

\vspace{.25cm}

\noindent \textbullet \quad \textbf{(The use of multiple network measures.)} 
There is a broad literature that uses repeated noisy measures to estimate nonlinear models with errors in variables, e.g., \cite{li2002robust}, \cite{chen2005measurement} and \cite{hu2017identification} or unobserved heterogeneity, e.g., \cite{hu2008identification} and \cite{bonhomme2016non}.
\cite{Hu2018hidden} use repeated measurement to estimate a binary choice model with misclassification and social interactions. 
These papers typically apply mathematical tools such as deconvolution, and eigenvalue or LU decomposition. 

In contrast, we use the repeated measures in a different way, one that does not require any deconvolution or matrix decomposition. 
Focusing on linear social networks, we exploit the identifying power from repeated measures through a standard 2SLS in Section \ref{sec:use_2H}, and apply a constructive closed-form algebraic argument to recover the misclassification rates in Section \ref{subsubsec: 2H}. 
Moreover, our 2SLS estimators are unlikely to suffer from weak instrument issues, because Assumption (A2) ensures correlation between $H$ and $G$, and our instruments are constructed from $H$.

An alternative way to use multiple network measures would be to consolidate them into a single measure, such as by recording a link if it exists according to either measure. 
This intuitive approach can lower, but not complete remove, the probability of missing links; therefore, it does not resolve the endogeneity bias due to misclassification. 
Besides, with two-sided misclassification, such consolidation rule increases the chance of imputing non-existent links, which aggravates the endogeneity bias. 


\section{Adaptation to Linear-in-Means Models} \label{sec:LIM}

In a linear-in-means (LIM) model, the adjacency matrix $G^*$ consists of binary entries $G^*_{ij}\in\{0,1\}$ (with $G^*_{ii}=0$ by convention) while the matrix $G$ in the actual structural form $y=\lambda G y + X\beta + \varepsilon$ is row-normalized, i.e., $G_{ij} \equiv G^*_{ij}/\left(\sum_{k} G^*_{ik}\right)$ if $\sum_{k} G^*_{ik}\neq 0$, and $G_{ij}\equiv 0$ otherwise. 
Suppose a researcher observes a noisy proxy $H^*$, with each off-diagonal entry $H^*_{ij}$ being potential misclassification of $G^*_{ij}$ satisfying (A1)-(A3), and $H^*_{ii}=0$ for all $i$.

To address link misclassification in these LIM models, we employ the same logic for the linear-in-sums model above, with necessary adaptation in how we adjust the noisy network measures to restore the exogeneity in $X$. 
Namely, we recover misclassification rates $p_0,p_1$ from the noisy measures $H^*$ as in Section \ref{sec:get_p}, and use them to construct a new transformation, denoted by $\widetilde W(H^*;p_0,p_1)$, so that $E(\widetilde W\vert X,G) = G$. 
In contrast with the linear-in-sums model in Section \ref{sec:augmentBias}, this new transformation requires a qualitatively different idea, because the row normalized $G$ in the structural form is \textit{non-linear} in the actual adjacency $G^*$. 

Let $H^*_{i\cdot}$ be the $i$-th row in $H^*$, and define $G^*_{i\cdot}$, $H_{i\cdot}$, $G_{i\cdot}$ similarly. 
It suffices to consider row-specific transformation.\footnote{
    Recall $G^*_{ii}=G_{ii}=0$ and $H^*_{ii}=H_{ii}=0$ by construct. So, we focus on off-diagonal entries only.} That is, for each row $i$, our objective is to construct a vector-valued function $\widetilde{W}_{i\cdot}(\cdot;p_0,p_1)$ that maps from $H^*_{i\cdot}\in \{0\}\times\{0,1\}^{n-1}$ to $\{0\}\times\mathbb R^{n-1}$ (where $n$ is the network size and $\widetilde{W}_{ii}(\cdot;p_0,p_1)$ is degenerate at zero by construct), so that  
\begin{equation} \label{eq:LS}
    E[\widetilde{W}_{ij}(H^*_{i\cdot};p_0,p_1)\vert X,G_{i\cdot}]=G_{ij} \text{ for all } j\neq i \text{ and all realized }G_{i\cdot}.
\end{equation}

There is one-to-one mapping between $G^*_{i\cdot}$ and its row-normalization $G_{i\cdot}$. Hence the conditional probability mass $P(H^*_{i\cdot}\vert G_{i\cdot})$ is identical to $P( H^*_{i\cdot} \vert G^*_{i\cdot})$, which, under the conditions (A1) and (A2), is determined by the misclassification rates $(p_0,p_1)$ as follows: 
\begin{equation} \label{eq:CPM}
    P(H^*_{i\cdot}\vert G^*_{i\cdot}) = \prod_{j\neq i} \left[(1-p_1)^{H^*_{ij}G^*_{ij}}(p_1)^{(1-H^*_{ij})G^*_{ij}}(p_0)^{H^*_{ij}(1-G^*_{ij})}(1-p_0)^{(1-H^*_{ij})(1-G^*_{ij})}\right].
\end{equation}
Consequently, for any $j\neq i$, the equalities in (\ref{eq:LS}) form a linear system of equations:
\begin{equation} \label{eq:LS2}
    \mathbf P_{(H^*_{i\cdot}\vert G^*_{i\cdot})} \times \widetilde W_{ij} = \widetilde V_{ij},
\end{equation}
where $\widetilde W_{ij}$ is $(2^{n-1})$-by-1 with components indexed by the realization of $H^*_{i\cdot}$, and $\mathbf P_{(H^*_{i\cdot}\vert G^*_{i\cdot})}$ is $(2^{n-1})$-by-$(2^{n-1})$ with rows indexed by the realizations of $G^*_{i\cdot}$ and columns indexed by $H^*_{i\cdot}$. 

\begin{proposition} \label{pn:generalCPM}
    The $\mathbf P_{(H^*_{i\cdot}\vert G^*_{i\cdot})}$ in (\ref{eq:LS2}) is the $(n-1)$-th Kronecker power of \[
        \left(\begin{array}{cc}
            1-p_0 & p_0 \\ 
            p_1 & 1-p_1 
        \end{array} \right),
    \]
    and is non-singular whenever $p_0+p_1\neq 1$.
\end{proposition}

The proof uses induction, and is included in the Online Appendix. The $2^{n-1}$-vector $\widetilde V_{ij} $ on the right-hand side of (\ref{eq:LS2}) has each component indexed by a realization of $G^*_{i\cdot}$ and equal to $G^*_{ij}/(\sum\nolimits_{k}G^*_{ik})\equiv G_{ij}$ (or $0$ if $\sum\nolimits_{k}G^*_{ik}=0$).
For each pair $j\neq i$, it is known by enumerating all elements in the support of $G^*_{i\cdot}$.
Thus, we can uniquely solve for the vector $\widetilde{W}_{ij}=\mathbf P_{(H^*_{i\cdot}\vert G^*_{i\cdot})}^{-1}\times\widetilde V_{ij}$ from (\ref{eq:LS2}), because $ \mathbf P_{(H^*_{i\cdot}\vert G^*_{i\cdot})} $ is invertible from $p_0+p_1<1$ in (A2).

By construct, $\widetilde W_{ij}$ is a transformation of $H^*_{i\cdot}$ using $(p_0,p_1)$, and satisfies (\ref{eq:LS}).
While constructing $\widetilde{W}_{ij}$, we enumerate all elements in the support of $G^*_{i\cdot}$. 
Yet this transformation is not a function of the actual $G^*_{i\cdot}$ that underlies the observed $H^*_{i\cdot}$. 
We apply this approach to construct all off-diagonal components in the matrix $\widetilde{W}(H^*;p_0,p_1)$.

\noindent \textbf{Example 1.}
To illustrate the idea, consider a simple case with $n=3$, and focus on the first row of the network, i.e., $G_{1.}^{\ast }$ and $H_{1.}^{\ast }$, for now. 
Since $G_{11}^{\ast }=H_{11}^{\ast }\equiv 0$, we have $2^{n-1}=4$ realizations of all possible $G_{1.}^{\ast }$ and $H_{1.}^{\ast}$. 
The probability mass for $H^*_{i\cdot}$ conditional on $G^*_{i\cdot}$ is:


\ \ \ 

\begin{tabular}{|l|l|l|l|l|}
\hline
$\mathbf{P}(H_{1.}^{\ast }|G_{1.}^{\ast })$ & $H_{1.}^{\ast }=$(0,0,0) & $%
H_{1.}^{\ast }=$(0,0,1) & $H_{1.}^{\ast }=$(0,1,0) & $H_{1.}^{\ast }=$(0,1,1)
\\ \hline
$G_{1.}^{\ast }=$(0,0,0) & $(1-p_{0})^{2}$ & $(1-p_{0})p_{0}$ & $%
p_{0}(1-p_{0})$ & $p_{0}^{2}$ \\ \hline
$G_{1.}^{\ast }=$(0,0,1) & $(1-p_{0})p_{1}$ & $(1-p_{0})(1-p_{1})$ & $%
p_{0}p_{1}$ & $p_{0}(1-p_{1})$ \\ \hline
$G_{1.}^{\ast }=$(0,1,0) & $p_{1}(1-p_{0})$ & $p_{1}p_{0}$ & $%
(1-p_{1})(1-p_{0})$ & $(1-p_{1})p_{0}$ \\ \hline
$G_{1.}^{\ast }=$(0,1,1) & $p_{1}^{2}$ & $p_{1}(1-p_{1})$ & $(1-p_{1})p_{1}$
& $(1-p_{1})^{2}$ \\ \hline
\end{tabular}

\ \ \ 

\noindent For this case, $\mathbf P_{(H^*_{i\cdot}\vert G^*_{i\cdot})}$ on the left-hand side of (\ref{eq:LS2}) is the matrix in the 4-by-4 table above. 

For $(i,j)=(1,2)$, the right-hand side of (\ref{eq:LS2}) is $\widetilde V_{12}=(0,0,1,1/2)'$.
This is because $G_{12}$ is $0$ when $G^*_{1\cdot}=(0,0,0)\text{ or }(0,0,1)$, $G_{12}=1$ when $G^*_{1\cdot}=(0,1,0)$, and $G_{12}=1/2$ when $G^*_{1\cdot}=(0,1,1)$.
Likewise, for $(i,j)=(1,3)$, the right-hand side of (\ref{eq:LS2}) is $\widetilde V_{13}=(0,1,0,1/2)'$. 

It then follows from the linear system (\ref{eq:LS2}) that
\[
\widetilde{W}_{12}=\mathbf P_{(H^*_{1\cdot}\vert G^*_{1\cdot})}^{-1}\times\widetilde V_{12}=\frac{1}{(1-p_{0}-p_{1})^{2}}\left( 
\begin{array}{c}
\frac{1}{2}p_{0}^{2}-p_{0}(1-p_{1}) \\ 
p_{0}p_{1}-\frac{1}{2}p_{0}(1-p_{0}) \\ 
(1-p_{1})(1-p_{0})-\frac{1}{2}p_{0}(1-p_{0}) \\ 
\frac{1}{2}(1-p_{0})^{2}-(1-p_{0})p_{1}%
\end{array}%
\right); 
\]
and $\widetilde{W}_{13}$ is similar to $\widetilde{W}_{12}$, only swapping the order of the second and third components.

For concreteness, consider a group in the sample, indexed by $s$, whose first row is $H^*_{s,1\cdot} = (0,1,1)$. 
The transformation for this row is $(0,\widetilde W_{12,4},\widetilde W_{13,4})$, where $\widetilde W_{1j,4}$ is the fourth component in the solution $\widetilde W_{1j}$ for $j\in\{2,3\}$.
(Recall $ H^*_{s,1\cdot} = (0,1,1) $ corresponds to the fourth column in $\mathbf P_{(H^*_{i\cdot}\vert G^*_{i\cdot})}$ above.)
By the same token, for another group $t$ in the sample with $H^*_{t,1\cdot} = (0,0,1)$, the transformed first row is $ (0, \widetilde W_{12,2}, \widetilde W_{13,2})$; 
and so on and so forth.  

We use the same logic to transform of all other rows. 
For instance, for the second row, we construct $\mathbf P(H^*_{2\cdot}\vert G^*_{2\cdot})$ and $\widetilde V_{21},\widetilde V_{23} \in \mathbb R^4$ similarly using $p_0,p_1$. 
    In fact, 
    these are \textit{identical} to $\mathbf P(H^*_{1\cdot}\vert G^*_{1\cdot})$ and $\widetilde V_{12},\widetilde V_{13}$ respectively, if the the support of $G^*_{2\cdot}$ (and $H^*_{2\cdot}$) are numerated in the order of (0,0,0), (0,0,1), (1,0,0), (1,0,1). 
In this case, for a group $s$ with $ H^*_{s,2\cdot} = (1,0,1)$, the transformed second row is $(\widetilde W_{21,4},0,\widetilde W_{23,4})$ where $\widetilde W_{2j,4}$ is the fourth component in the solution  $\widetilde W_{2j}=\mathbf P_{(H^*_{2\cdot}\vert G^*_{2\cdot})}^{-1} \widetilde V_{2j}$ for $j \in \{1,3\}$. 
\hfill (\textbf{End of Example 1.})

\ \ \ 

In summary, to handle link misclassification in LIM models, we construct a new adjusted structural form, similar to (\ref{SF1}) in Section \ref{sec:augmentBias}, only with $W(\cdot)$ replaced by $\widetilde{W}(\cdot)$.
Then, with (\ref{eq:LS}) holding by construction, we can apply the subsequent steps from Section \ref{sec:augmentBias} and \ref{sec:use_H} to identify the social effects in the linear in means model with link misclassification.


\section{Two-Step Estimation} \label{sec:estimation}

We now propose adjusted 2SLS estimators for the coefficients of structural effects $(\lambda ,\beta^{\prime })^{\prime }$. Consider a sample of $S$ independent groups with each group $s$ consisting of $n_{s}$ members. For each group $s$, the sample reports an $n_{s}$-by-$1$ vector of outcomes $y_{s}$, an $n_{s}$-by-$K$ matrix of regressors $X_{s}$, and either an $n_{s}$-by-$n_{s}$ unsymmetrized measure $H_{s}$, or two $n_{s} $-by-$n_{s}$ conditionally independent symmetrized measures $H_{s}^{(1)}$ and $H_{s}^{(2)}$.

\subsection{Closed-form estimation of misclassification rates} \label{subsec:est_MR}

To estimate misclassification rates, we apply the analog principle to the identification in Section \ref{sec:get_p}. We include closed-form estimates for completeness; the logic for these estimators is self-evident as presented in Section \ref{sec:get_p} and detailed proof is in Appendix A2.

First, consider the case in Section \ref{subsubsec: 2H}, where the sample reports two conditionally independent measures $H_s^{(1)}$ and $H_s^{(2)}$. To exploit identifying power from their joint distribution, let $H_{s,ij}^{(3)}\equiv \max \left\{H_{s,ij}^{(1)},H_{s,ij}^{(2)}\right\} $ for each $(i,j)$-th entry in $H_s^{(t)}$. 
For $t=1,2,3$, define $\hat{\psi}_{1}^{(t)}$:%
\begin{equation}
\hat{\psi}_{1}^{(t)}\equiv \frac{\sum_{s}\left[ \frac{1}{n_{s}(n_{s}-1)}\left(
\sum_{i\not=j}H_{s,ij}^{(t)}1\{\phi _{s,ij}=1\}\right) \right] 
}{\sum_{s}\left[ \frac{1}{n_{s}(n_{s}-1)}\left( \sum_{i\not=j}1\{\phi
_{s,ij}=1\}\right) \right] }\text{,}  \label{defn:psi_ta}
\end{equation}%
where $\phi _{s,ij}$ is short for $\phi _{ij}(X_{s})$. Define $\hat{\psi}_{0}^{(t)}$ by
replacing $\phi _{s,ij}=1$ with $\phi _{s,ij}=0$ in (\ref{defn:psi_ta}).

For instance, in our application, we define $\phi _{ij}(X_{s})$ as $1\{X_{s,i,k}=X_{s,j,k}\}$, where $X_{s,i,k}$ is the $k$-th component in $X_{s,i}$ that reports the individual $i$'s caste. 
In this case, $\hat{\psi}^{(t)}_1$ and $\hat{\psi}^{(t)}_0$ are, respectively, the fraction of same-caste and different-caste pairs that are linked according to the measures $H_s^{(t)}$ for $t=1,2,3$. Using the sample moments, we define:%
\begin{eqnarray*}
\widehat{\mathcal{C}_{2}} &\equiv &\frac{\hat{\psi}_{0}^{(1)}-\hat{\psi}_{1}^{(1)}}{%
\hat{\psi}_{0}^{(2)}-\hat{\psi}_{1}^{(2)}}\text{, \space  } 
\widehat{\mathcal{C}_{1}} \equiv \hat{\psi}_{1}^{(1)}-1+\frac{\hat{\psi}_{0}^{(3)}-%
\hat{\psi}_{1}^{(3)}}{\hat{\psi}_{0}^{(2)}-\hat{\psi}_{1}^{(2)}}-(1-\hat{\psi}_{1}^{(2)})%
\widehat{\mathcal{C}_{2}}\text{, } \\
\widehat{\mathcal{C}_{0}} &\equiv &\hat{\psi}_{1}^{(1)}+\hat{\psi}_{1}^{(2)}-\hat{\psi}%
_{1}^{(1)}\hat{\psi}_{1}^{(2)}-\hat{\psi}_{1}^{(3)}\text{, \space }\hat{\xi}\equiv (2\widehat{\mathcal{C}_{2}})^{-1}\left( \widehat{\mathcal{C}%
_{1}}+\sqrt{\left( \widehat{\mathcal{C}_{1}}\right) ^{2}+4\widehat{\mathcal{C%
}_{2}}\widehat{\mathcal{C}_{0}}}\right). 
\end{eqnarray*}
\noindent Our closed-form estimators for misclassification rates are then:%
\begin{equation*}
\hat{p}_{0}^{(1)}\equiv \hat{\psi}_{1}^{(1)}-\widehat{\mathcal{C}_{2}}\hat{\xi}%
\text{, \space  }\hat{p}_{0}^{(2)}\equiv \hat{\psi}_{1}^{(2)}-\hat{\xi}\text{, and } \hat{p}_{1}^{(t)}\equiv 1-\hat{p}_{0}^{(t)}-\frac{\hat{\psi}_{1}^{(t)}-\hat{p}%
_{0}^{(t)}}{\hat{\pi}_{1}}\text{ for }t=1,2\text{, where}
\end{equation*}%
\begin{equation*}
\hat{\pi}_{1}=\frac{\left( \hat{\psi}_{1}^{(1)}-\hat{p}_{0}^{(1)}\right) \left( 
\hat{\psi}_{1}^{(2)}-\hat{p}_{0}^{(2)}\right) }{\left( 1-\hat{p}_{0}^{(1)}\right)
\left( \hat{\psi}_{1}^{(2)}-\hat{p}_{0}^{(2)}\right) +\left( 1-\hat{p}%
_{0}^{(2)}\right) \left( \hat{\psi}_{1}^{(1)}-\hat{p}_{0}^{(1)}\right) -\left( 
\hat{\psi}_{1}^{(3)}-\hat{p}_{0}^{(3)}\right) }\text{,}
\end{equation*}%
with $\hat{p}_{0}^{(3)}\equiv \hat{p}_{0}^{(1)}+\hat{p}_{0}^{(2)}-\hat{p}%
_{0}^{(1)}\hat{p}_{0}^{(2)}$ by construction.

Next, consider the case in Section \ref{subsubsec: 1H}, where the sample reports a single, unsymmetrized measure $H$ with misclassification rates $(p_{0},p_{1})$ while the actual $G$ is known to be symmetric. 
Estimation of $(p_0,p_1)$ in this case follows from almost identical steps. 
For any \textit{unordered} pair $\{i,j\}$, define $H_{s,\{i,j\}}^{(1)}\equiv H_{s,ij}$ and $H_{s,\{i,j\}}^{(2)}\equiv H_{s,ji}$. By construction, $p_{1}^{(t)}=p_{1}$ and $p_{0}^{(t)}=p_{0}$ do not vary between $t=1,2$. 
Construct $H_{\{i,j\}}^{(3)}=\max \{H_{\{i,j\}}^{(1)},H_{\{i,j\}}^{(2)}\}$; define $\hat{\psi}_{1}^{(t)}$ and $\hat{\psi}_{0}^{(t)}$ in this case by replacing $\frac{1}{n_{s}(n_{s}-1)}$, $\sum_{i\not=j}$ and $H_{s,ij}^{(t)}$ in (\ref{defn:psi_ta}) with $\frac{2}{n_{s}(n_{s}-1)}$, $\sum_{i>j}$, and $H_{s,\{i,j\}}^{(t)}$ respectively.
Replace $\widehat{\mathcal{C}_{2}}$ with $1$, and replace $\hat{\psi}_{1}^{(1)},%
\hat{\psi}_{1}^{(2)}$ with their average in $\widehat{\mathcal{C}_{1}},\widehat{\mathcal{C}_{0}}$ and all subsequent expressions. These lead to a single pair of estimates $(\hat{p}_{0},\hat{p}_{1})$.

We derive the limiting distribution of these estimators using a standard delta method. Consider the case with a single unsymmetrized measure in Section \ref{subsubsec: 1H}. 
For each group $s$, define $\upsilon_{1s,1}\equiv \frac{2}{n_{s}(n_{s}-1)}\sum_{i>j}H_{s,\{i,j\}}1\{\phi_{s,\{i,j\}}=1\}$ and $\upsilon _{2s,1}\equiv \frac{2}{n_{s}(n_{s}-1)}\sum_{i>j}1\{\phi _{s,\{i,j\}}=1\}$; define $\upsilon _{1s,0}$, $\upsilon _{2s,0}$ analogously be replacing $\phi_{s,\{i,j\}}=1$ with $\phi_{s,\{i,j\}}=0$. 
Let $\upsilon_{s}\equiv (\upsilon _{1s,1},\upsilon _{2s,1},\upsilon_{1s,0},\upsilon _{2s,0})^{\prime }$. 
The estimator $\hat{p}=(\hat{p}_{0},\hat{p}_{1})$ is a closed-form function of $\upsilon _{s}$; it has an asymptotic linear presentation: $\sqrt{S}(\widehat{p}-p)=\tfrac{1}{\sqrt{S}}\sum\nolimits_{s}\underset{\equiv
\tau _{s}}{\underbrace{\mathcal{J}_{0}\times \left[ \upsilon _{s}-E(\upsilon
_{s})\right] }}+o_{p}(1)$, where $\mathcal{J}_{0}$ denotes the Jacobian matrix of $\hat{p}$ w.r.t. the sample averages of $\upsilon _{s}$, evaluated at population counterparts. 
Thus $\sqrt{S}(\widehat{p}-p)$ converges in distribution to a multivariate normal distribution with zero means and a covariance matrix $E(\tau _{s}\tau _{s}^{\prime })$. 
Limiting distribution for the case with two measures in Section \ref{subsubsec: 2H} follows from the same type of arguments.

\subsection{Adjusted 2SLS using a single unsymmetrized measure} \label{subsec:A2SLS_1H}
First consider the setting in Section \ref{subsec:use_1H}, where the sample reports a single \textit{unsymmetrized} measure $H_{s}$ for each group. 
Let $p\equiv \left( p_{0},p_{1}\right) ^{\prime }$. For each group $s$, define $R_{s}(p)\equiv \left( W_{s}(p)y_{s},X_{s}\right)$ and $Z_{s}\equiv \left( H_{s}^{\prime }X_{s},X_{s}\right)$, where $W_{s}(p)\equiv \lbrack H_{s}-p_0(\iota _{s}\iota'
_{s}-I_{s})]/(1-p_{0}-p_{1})$. Let $N\equiv \sum_{s=1}^{S}n_{s}$, and $Y$ be
an $N$-by-1 vector that stacks $y_{s}$ for $s=1,...,S$. Let $\mathbf{R}(p)$
be an $N$-by-$(K+1)$ matrix that stacks $R_{s}(p)$, and $%
\mathbf{Z}$ an $N$-by-$2K$ matrix that stacks $Z_{s}$ for all $s$. Let  $\mathbf{A\equiv Z}^{\prime }\mathbf{R}(\widehat{p})\text{ and }\mathbf{B}%
\equiv \mathbf{Z}^{\prime }\mathbf{Z}\text{, with }\hat{p}\equiv \left( \hat{%
p}_{0},\hat{p}_{1}\right) ^{\prime }$. Our
adjusted 2SLS estimator for $\theta \equiv (\lambda ,\beta ^{\prime
})^{\prime }$ is:%
\begin{equation}
\widehat{\theta }\equiv \left( \mathbf{A}^{\prime }\mathbf{B}^{-1}\mathbf{A}%
\right) ^{-1}\mathbf{A}^{\prime }\mathbf{B}^{-1}\left( \mathbf{Z}^{\prime
}Y\right) \text{.}  \label{2sls}
\end{equation}%

We now present the limiting distribution of $\widehat{\theta }$ as $%
S\rightarrow \infty $. Define $\Sigma _{0}\equiv \left( A_{0}^{\prime }B_{0}^{-1}A_{0}\right)
^{-1}A_{0}^{\prime }B_{0}^{-1}$,
where %
$ A_{0} \equiv \lim_{S\rightarrow \infty }\tfrac{1}{S}\sum\nolimits_{s=1}^{S}E%
\left[ Z_{s}^{\prime }R_{s}(p)\right] \text{ and }B_{0}\equiv
\lim_{S\rightarrow \infty }\tfrac{1}{S}\sum\nolimits_{s=1}^{S}E(Z_{s}^{%
\prime }Z_{s})\text{.}$
For each group $s$ and individual $i\leq n_{s}$, let $R_{s,i}(p)$ denote the
corresponding row in $\mathbf{R}(p)$, and $\triangledown _{p}R_{s,i}(p)$ be
the $(K+1)$-by-$2$ Jacobian of $R_{s,i}(p)$ with respect to $p$.%
\footnote{%
    The last $K$ rows in $\triangledown _{p}R_{s,i}(p)$ are 0; its 1st row is the $i$-th row in $\left( \frac{{H}_{s}-(1-p_1)(\iota_s \iota'_s -I_{s})}{(1-p_{0}-p_{1})^2}y_{s},\frac{{H}_{s}-p_0(\iota_s \iota'_s -I_{s})}{(1-p_{0}-p_{1})^2}y_{s}\right) $.} 

Let $\triangledown _{p}\left[R_{s}(p)\theta \right] $ be an $n_{s}$-by-$2$ matrix with each row $i\leq n_{s}$ being $\theta ^{\prime }\triangledown _{p}R_{s,i}(p)$; let $\triangledown _{p}\left[ \mathbf{R}(p)\theta \right] $ be an $N$-by-$2$ matrix formed by stacking these $n_{s}$-by-$2$\ matrices over $s=1,2,...,S$.
Define $\kappa _{s}\equiv Z_{s}^{\prime }v_{s}-F_{0}\tau _{s}$, where $v_{s}$ is the $n_{s}$-by-1 vector of composite errors in the feasible structural form (\ref%
{SF1}), and $F_{0}\equiv \lim_{S\rightarrow \infty
}S^{-1}\sum\nolimits_{s=1}^{S}E\left\{ Z_{s}^{\prime }\triangledown \left[
R_{s}(p)\theta \right] \right\} $. Intuitively, $F_{0}$ illustrates how the moment condition in this adjusted 2SLS depends on misclassification rates $p$, and the added term ``$-F_{0}\tau _{s}$'' in the influence function accounts for the first-stage estimation error in 
$\widehat{p}$.

\begin{proposition} \label{pn:Asy_2SLS}
Suppose (A1)-(A4) hold, and (IV-R) is satisfied with $%
Z\equiv (H^{\prime }X,X)$. Then under the regularity conditions (REG) in the Online Appendix,
\begin{equation*}
\sqrt{S}\left( \widehat{\theta }-\theta \right) \overset{d}{\longrightarrow }%
\mathcal{N}(0,\Sigma _{0}E(\kappa _{s}\kappa _{s}^{\prime })\Sigma
_{0}^{\prime })\text{.}
\end{equation*}%
\end{proposition}

Note that this limiting distribution includes group-level clustering. The conditions in (\textit{REG}), presented in the Online Appendix, are needed for
applying the law of large numbers, the central limit theorem, and the delta
method to observations from independent groups with heterogeneous sizes. Standard errors for $\widehat{\theta }$\ (which are clustered at the group level) are calculated by replacing $A_{0}$%
, $B_{0}$, $F_{0}$, and $E(\kappa _{s}\kappa _{s}^{\prime })$ with their
sample analogs: 
\begin{equation*}
\widehat{A}=\tfrac{1}{S}\sum\nolimits_{s}Z_{s}^{\prime }R_{s}(\widehat{p})%
\text{, \space }\widehat{B}=\tfrac{1}{S}\sum\nolimits_{s}Z_{s}^{\prime }Z_{s}\text{, \space }\widehat{\kappa }_{s}=Z_{s}^{\prime }\left( y_{s}-R_{s}(\widehat{p})%
\widehat{\theta }\right) -\widehat{F}\widehat{\tau }_{s}\text{.}
\end{equation*}

One could combine the two steps in Section \ref{subsec:est_MR} and \ref{subsec:A2SLS_1H} into a single GMM step by stacking the moments used in these two sections. 
This would allow one to estimate $\theta$ jointly with $(p_0,p_1)$, and standard GMM asymptotics could be applied. However, this GMM requires numerically solving a nonlinear optimization problem, while the two-step method yields a closed-form estimator that is straightforward to compute with no numerical searching, thus providing a computational advantage over the GMM alternative with numerical stability. 

\subsection{Adjusted 2SLS using multiple measures} \label{subsec: est_2H}

We now apply the same idea for estimation under the other setting in Section \ref{sec:use_2H}, where the sample reports two conditionally independent measures $H_{s}^{(t)}$ for $t=1,2$, with misclassification rates $p_{0}^{(t)},p_{1}^{(t)}$ for $t=1,2$ respectively.
These measures may either be \textit{symmetrized} or \textit{unsymmetrized}.
To reiterate, when $H_{s}^{(t)}$ are \emph{unsymmetrized}, our estimation method applies \emph{regardless of} whether the actual adjacency $G$ is symmetric or not; 
on the other hand, when $H_{s}^{(t)}$ are symmetrized, (A1) holds only if $G$ is symmetric.

As noted in Section \ref{sec:use_2H}, these measures lead to two feasible structural forms:%
\begin{equation}
y_{s}=R_{s}^{(t)}\theta +v_{s}^{(t)}\text{ for }t=1,2\text{,}  \label{eq:2SF}
\end{equation}%
where $\theta \equiv (\lambda ,\beta ^{\prime })^{\prime }$, $R_{s}^{(t)}\equiv \left( W_{s}^{(t)}y_{s},X_{s}\right) $ and $v_{s}^{(t)}\equiv \varepsilon _{s}+\lambda \left( G_{s}-W_{s}^{(t)}\right) y_{s}$, with $W_{s}^{(t)}\equiv \frac{H_{s}^{(t)}-p_0^{(t)}(\iota _{s}\iota' _{s}-I_{s})}{1-p_{0}^{(t)}-p_{1}^{(t)}}$. 
This leads to two sets of moment conditions:
\begin{eqnarray*}
    E\left[(H_s^{(3-t)}X,X)^{\prime}(y_s-\lambda W^{(t)}_s y_s - X_s\beta)\right] 
     &=& E\left[(H_s^{(3-t)}X,X)^{\prime}v^{(t)}_s\right] = 0 \text{ for }t=1,2,
\end{eqnarray*}
with instruments $Z_{s}^{(t)}\equiv \left(H_{s}^{(3-t)}X_{s},X_{s}\right)$ for $t=1,2$. 
Stack the moments by defining:%
\begin{equation*}
\tilde{Z}_{s}\equiv \left( 
\begin{array}{cc}
Z_{s}^{(1)} & 0 \\ 
0 & Z_{s}^{(2)}%
\end{array}%
\right) \text{; }\tilde{y}_{s}\equiv \left( 
\begin{array}{c}
y_{s} \\ 
y_{s}%
\end{array}%
\right) \text{; }\tilde{R}_{s}\equiv \left( 
\begin{array}{c}
R_{s}^{(1)} \\ 
R_{s}^{(2)}%
\end{array}%
\right) \text{.}
\end{equation*}%
Instrument exogeneity then implies: $E\left[ \tilde{Z}_{s}^{\prime }(\tilde{y}_{s}-\tilde{R}_{s}\theta )\right] =0$. This moment condition identifies $\theta $, provided $E(\tilde{Z}_{s}^{\prime }\tilde{R}_{s})$ has full rank. 
Using arguments similar to Proposition \ref{pn:IVR} in Section \ref{subsec:use_1H}, we can derive analogous sufficient conditions for this rank condition. 

We define a system, or stacked adjusted 2SLS (S2SLS) estimator as follows. 
Let $\mathbf{\tilde{Z}}$ be a $2N$-by-$4K$ matrix constructed by vertically stacking $S$ matrices $(\tilde{Z}_{s})_{s\leq S}$. 
Likewise, construct a $2N$-by-$(K+1)$ matrix$\ \mathbf{\tilde{R}}$ by stacking $(\tilde{R}_{s})_{s\leq S}$, where $p_{0}^{(t)}$ and $p_{1}^{(t)}$ are replaced by estimates $\hat{p}_{0}^{(t)}$ and $ \hat{p}_{1}^{(t)}$, and construct a $2N$-by-$1$ vector $\mathbf{\tilde{y}}$ by stacking $(\tilde{y}_{s})_{s\leq S}$. 
The S2SLS estimator is%
\begin{equation}
\tilde{\theta}\equiv \mathbf{[\tilde{R}}^{\prime }\mathbf{\tilde{Z}(\tilde{Z}%
}^{\prime }\mathbf{\tilde{Z})}^{-1}\mathbf{\tilde{Z}}^{\prime }\mathbf{%
\tilde{R}]}^{-1}\mathbf{\tilde{R}}^{\prime }\mathbf{\tilde{Z}(\tilde{Z}}%
^{\prime }\mathbf{\tilde{Z})}^{-1}\mathbf{\tilde{Z}}^{\prime }\mathbf{\tilde{%
y}}\text{.}  \label{eq:S2SLS}
\end{equation}%
This provides us with a single estimator that exploits both sets of instruments in the two structural forms in (\ref{eq:2SF}). 
Similar to $\hat{\theta}$ in (\ref{2sls}), we can construct the standard error for $\tilde{\theta}$ that accounts for estimation errors in $\hat{p}_{0}^{(t)},\hat{p}_{1}^{(t)}$ for $t=1,2$. We omit details here for brevity.

\section{Simulation} \label{sec:simulation}

We use the Monte Carlo simulation to examine the finite-sample performance of the adjusted 2SLS estimator proposed in Section \ref{sec:estimation}. The DGP is:%
\[
y_{s}=\lambda G_{s}y_{s}+X_{s}\beta +\alpha _{s}+\varepsilon _{s},\text{with 
}s=1,...,S. 
\]
\noindent Each member $i$ in group $s$ has individual characteristics $X_{s,i}\equiv(X_{s,i,1},X_{s,i,2})\in \mathbb{R}^{2}$, drawn independently across $i$ and $s$, from a Bernoulli with success probability 0.5 and an $N(0,1)$ respectively.
The error term $\varepsilon _{s,i}$ is also drawn from $N(0,1)$ independently across $i$ and $s$. 
The coefficients for social effects are $\lambda =0.05$ and $\beta =(\beta_{1},\beta _{2})=(1,2)$. The group-level fixed effect is $\alpha _{s}=5\overline{X}_{s}\beta -1.5+e_{s}$, where $\overline{X}_{s}$ is the group average of $X_{s}$ and $e_{s}$ is drawn from $N(0,1)$ independently across $i$ and $s$. 
This construction allows the fixed effects $\alpha _{s}$ to be correlated with group demographics $\overline X_{s}\beta $. 
The dyadic link formation rates are%
\begin{eqnarray*}
\pi _{1} =\Pr \{G_{s,ij}=1|X_{s,i,1}=X_{s,j,1}\}=0.2\text{ and }\pi _{0} =\Pr \{G_{s,ij}=1|X_{s,i,1}\neq X_{s,j,1}\}=0.1.
\end{eqnarray*}%
For $t=1,2$, we generate the following measure $H_{s}^{(t)}$ with link misclassification:
\[
H_{s,ij}^{(t)}=m_{ij,1}^{(t)}\cdot 1\{G_{s,ij}=1\}+(1-m_{ij,0}^{(t)})\cdot 1\{G_{s,ij}=0\}, 
\]%
where $m_{ij,0}^{(t)}$ and $m_{ij,1}^{(t)}$ are drawn independently across ordered pairs $(i,j)$ from Bernoulli distributions with success probabilities $1-p_{0}^{(t)}$ and $1-p_{1}^{(t)}\ $ respectively. 

To see how various estimators behave in the presence of misclassified links, we use two sets of misclassification rates.  In the first set, the misclassification rates are \textit{small}:%
\begin{eqnarray*}
p_{0}^{(1)} &=&\Pr \{H_{s,ij}^{(1)}=1|G_{s,ij}=0,X\}=0.10,\text{ } 
p_{1}^{(1)} =\Pr \{H_{s,ij}^{(1)}=0|G_{s,ij}=1,X\}=0.20; \\
p_{0}^{(2)} &=&\Pr \{H_{s,ij}^{(2)}=1|G_{s,ij}=0,X\}=0.08,\text{ } 
p_{1}^{(2)} =\Pr \{H_{s,ij}^{(2)}=0|G_{s,ij}=1,X\}=0.16.
\end{eqnarray*}
In the second set, we specify \textit{large} misclassification rates that are twice as high: $p_{0}^{(1)} =0.20,\text{ } p_{1}^{(1)} =0.40;\text{ } p_{0}^{(2)}=0.16,\text{ } 
p_{1}^{(2)} =0.32$. Each group has the same size $n_{s}=n$. We experiment with group sizes $n\in\{25,50\}$ and the number of groups $S\in\{50,100\}$. The total sample size is $nS.$
For each combination of $\{n,S\},$ we generate $Q=100$ samples. 

Table 1(a) reports the mean and the standard deviation of the estimates for $\pi_0, \pi_1$, $p_0^{(1)}, p_1^{(1)}$, $p_0^{(2)}, p_1^{(2)}$ based on their empirical distribution across these $100$ samples.

\newpage
\begin{center}
 \renewcommand{\arraystretch}{0.7} 
\textbf{Table 1(a): Estimates of Misclassification Rates and Network Parameters}

\bigskip
\begin{tabular}{l|llllll}
\hline \hline
Small & $\pi _{1}=$0.2 & $\pi _{0}=$0.1 & $p_{0}^{(1)}=$0.1 & $p_{1}^{(1)}=$%
0.2 & $p_{0}^{(2)}=$0.08 & $p_{1}^{(2)}=$0.16 \\ \hline
$S=50$ & \multicolumn{1}{|c}{$\widehat{\pi }_{1}$} & \multicolumn{1}{c}{$%
\widehat{\pi }_{0}$} & \multicolumn{1}{c}{$\widehat{p}_{0}^{(1)}$} & 
\multicolumn{1}{c}{$\widehat{p}_{1}^{(1)}$} & \multicolumn{1}{c}{$\widehat{p}%
_{0}^{(2)}$} & \multicolumn{1}{c}{$\widehat{p}_{1}^{(2)}$} \\ 
$n=25$ & 0.2009 & 0.1015 & 0.0990 & 0.2020 & 0.0792 & 0.1638 \\ 
& (0.0123) & (0.0081) & (0.0061) & (0.0301) & (0.0059) & (0.0349) \\ 
$n=50$ & 0.1996 & 0.0998 & 0.1002 & 0.2000 & 0.0800 & 0.1573 \\ 
& (0.0063) & (0.0042) & (0.0031) & (0.0150) & (0.0031) & (0.0186) \\ 
\hline
$S=100$ &  &  &  &  &  &  \\ 
$n=25$ & 0.1994 & 0.0997 & 0.0996 & 0.1968 & 0.0804 & 0.1588 \\ 
& (0.0099) & (0.0060) & (0.0042) & (0.0241) & (0.0047) & (0.0245) \\ 
$n=50$ & 0.2006 & 0.1006 & 0.0997 & 0.2011 & 0.0798 & 0.1608 \\ 
& (0.0043) & (0.0029) & (0.0020) & (0.0099) & (0.0019) & (0.0112) \\ 
\hline
Large & $\pi _{1}=$0.2 & $\pi _{0}=$0.1 & $p_{0}^{(1)}=$0.2 & $p_{1}^{(1)}=$%
0.4 & $p_{0}^{(2)}=$0.16 & $p_{1}^{(2)}=$0.32 \\ \hline
$S=50$ & \multicolumn{1}{|c}{$\widehat{\pi }_{1}$} & \multicolumn{1}{c}{$%
\widehat{\pi }_{0}$} & \multicolumn{1}{c}{$\widehat{p}_{0}^{(1)}$} & 
\multicolumn{1}{c}{$\widehat{p}_{1}^{(1)}$} & \multicolumn{1}{c}{$\widehat{p}%
_{0}^{(2)}$} & \multicolumn{1}{c}{$\widehat{p}_{1}^{(2)}$} \\ 
$n=25$ & 0.2032 & 0.1039 & 0.1994 & 0.4012 & 0.1586 & 0.3191 \\ 
& (0.0370) & (0.0260) & (0.0092) & (0.0442) & (0.0112) & (0.0654) \\ 
$n=50$ & 0.1987 & 0.0994 & 0.2005 & 0.3990 & 0.1602 & 0.3137 \\ 
& (0.0174) & (0.0122) & (0.0045) & (0.0224) & (0.0052) & (0.0330) \\ 
 \hline
$S=100$ &  &  &  &  &  &  \\ 
$n=25$ & 0.1987 & 0.0993 & 0.1995 & 0.3943 & 0.1604 & 0.3142 \\ 
& (0.0257) & (0.0173) & (0.0062) & (0.0322) & (0.0075) & (0.0452) \\ 
$n=50$ & 0.2011 & 0.1012 & 0.1998 & 0.4013 & 0.1594 & 0.3189 \\ 
& (0.0123) & (0.0090) & (0.0032) & (0.0159) & (0.0039) & (0.0216) \\ 

\hline \hline
\end{tabular}

Note: standard deviations based on 100 simulated samples are reported in parentheses.

\end{center}

From Table 1(a), we can see the misclassification rates $(p^{(t)}_0,p^{(t)}_1)$, as well as the network parameters $(\pi_0,\pi_1)$, are accurately estimated in all settings. 
For a fixed group number $S$, the standard deviations decrease at the rate $n$. 
For a fixed groups size $n$, the standard deviations decrease at the rate $\sqrt{S}$, as the size of the sample used for estimation is  $S\times n^2$. 
The standard deviations of these estimates are also larger when the misclassification rates are higher.

Then, we compare five estimators based on three versions of 2SLS estimation:  naive,  adjusted, and oracle (infeasible).
The naive 2SLS uses the noisy measure $H$ in place of the true network $G$, which means it uses $H_{s}y_{s}$ as an endogenous regressor and $H_{s}X_{s}$ as its instrument. 
The adjusted 2SLS estimator is what we proposed in Section \ref{sec:estimation}.
It requires two steps. First, estimate the misclassification rates based on $(H^{(1)},H^{(2)},X)$. Second, construct 
\[
W_{s}^{(t)}=\frac{H_{s}^{(t)}-\widehat{p}_{0}^{(t)}(\iota_{n}\iota'_{n}-I_{n})}{1-\widehat{p}_{0}^{(t)}-\widehat{p}_{1}^{(t)}}\text{ for }t=1,2, 
\]%
based on the first-step estimates $\widehat{p}_{0}^{(t)}$ and $\widehat{p}_{1}^{(t)}$, and apply 2SLS using $W_{s}^{(t)}y$ as an endogenous regressor and $W_{s}^{(t^{\prime})}X$ as its instrument where $t\neq t^{\prime }$. 
The oracle (infeasible) 2SLS uses the true $G_{s}y_{s}$, as an endogenous regressor, and uses $G_{s}X_{s}$ as its instrument.

Across the simulated samples indexed by $q=1,2,...,Q$, we record the empirical distribution of these estimates of $(\lambda ,\beta _{1},\beta _{2})$. 
Tables 1(b) and (c) report the average estimates, and standard deviations from this empirical distribution under different misclassification rates.

Simulation results in Tables 1(b) and (c) demonstrate the following patterns. 
First, the naive method ignoring the misclassification in $H$ has serious bias when estimating the peer effects $\lambda =0.05$. 
With lower misclassification rates, it estimates $\lambda $ at around 0.028 using $H^{(1)}$ and around 0.031 using $H^{(2)}$, about $40\%$ bias; with higher misclassification rates, it estimates $\lambda $ at around 0.013 using $H^{(1)}$ and around 0.018 using $H^{(2)}$, about $70\%$ bias. 
When estimating $\beta $, the naive estimation also shows bias, but relatively smaller than the bias in $\lambda$. 

Second, our proposed adjusted 2SLS can estimate $(\lambda,\beta _{1},\beta _{2})$ with high accuracy. The average estimates are very close to the oracle estimates, albeit with larger standard deviations. This is of course due to the noise from link misclassification as well as estimation errors in the initial estimates of the misclassification rates.


Third, with the fixed group size $n$, as the group number increases from $S=50$ to $100$, the standard deviation decreases by around 1/$\sqrt{2}$, consistent with our theory of $\sqrt{S}$ asymptotics.

\newpage 

\renewcommand{\arraystretch}{0.7}

{\centering
\textbf{Table 1(b): Peer Effects Estimation: Small Misclassification}

\begin{tabular}{c|cc|cc|c|cc|cc|c}
\hline \hline
& \multicolumn{5}{|c|}{${\small S=50}$} & \multicolumn{5}{|c}{${\small S=100}
$} \\ \hline
& \multicolumn{2}{|c}{\small Naive} & \multicolumn{2}{|c|}{\small Adjusted}
& {\small Oracle} & \multicolumn{2}{|c}{\small Naive} & \multicolumn{2}{|c|}%
{\small Adjusted} & {\small Oracle} \\ 
{\small Reg.} & ${\small H}^{(1)}{\small y}$ & ${\small H}^{(2)}{\small y}$
& $W^{(1)}{\small y}$ & $W^{(2)}{\small y}$ & ${\small Gy%
}$ & ${\small H}^{(1)}{\small y}$ & ${\small H}^{(2)}{\small y}$ & $W^{(1)}{\small y}$ & $W^{(2)}{\small y}$ & ${\small Gy}$ \\ 
{\small IV} & ${\small H}^{(1)}{\small X}$ & ${\small H}^{(2)}{\small X}$ & $%
{\small H}^{(2)}{\small X}$ & ${\small H}^{(1)}{\small X}$ & ${\small GX}$ & 
${\small H}^{(1)}{\small X}$ & ${\small H}^{(2)}{\small X}$ & ${\small H}%
^{(2)}{\small X}$ & ${\small H}^{(1)}{\small X}$ & ${\small GX}$ \\ \hline
${\small n=25}$ & \multicolumn{10}{c}{\small Expected \# of peers 3.75} \\ 
\hline
${\small \lambda =0.05}$ & {\small 0.0259} & {\small 0.0307} & {\small 0.0490%
} & {\small 0.0467} & {\small 0.0508} & {\small 0.0283} & {\small 0.0324} & 
{\small 0.0517} & {\small 0.0511} & {\small 0.0489} \\ 
{\small s.t.d} & {\small (0.007)} & {\small (0.006)} & {\small (0.012)} & 
{\small (0.014)} & {\small (0.005)} & {\small (0.005)} & {\small (0.005)} & 
{\small (0.008)} & {\small (0.009)} & {\small (0.007)} \\ 
${\small \beta }_{1}{\small =1}$ & {\small 1.0613} & {\small 1.0523} & 
{\small 1.0113} & {\small 1.0131} & {\small 1.0108} & {\small 1.0614} & 
{\small 1.0540} & {\small 1.0102} & {\small 1.0117} & {\small 1.0112} \\ 
{\small s.t.d} & {\small (0.078)} & {\small (0.081)} & {\small (0.079)} & 
{\small (0.086)} & {\small (0.062)} & {\small (0.064)} & {\small (0.066)} & 
{\small (0.062)} & {\small (0.064)} & {\small (0.078)} \\ 
${\small \beta }_{2}{\small =2}$ & {\small 1.9978} & {\small 1.9983} & 
{\small 1.9950} & {\small 1.9951} & {\small 2.0018} & {\small 2.0064} & 
{\small 2.0058} & {\small 2.0041} & {\small 2.0027} & {\small 1.9946} \\ 
{\small s.t.d} & {\small (0.046)} & {\small (0.046)} & {\small (0.047)} & 
{\small (0.047)} & {\small (0.031)} & {\small (0.032)} & {\small (0.032)} & 
{\small (0.034)} & {\small (0.032)} & {\small (0.046)} \\ \hline
${\small n=50}$ & \multicolumn{10}{c}{\small Expected \# of peers 7.5} \\ 
\hline
${\small \lambda =0.05}$ & {\small 0.0274} & {\small 0.0312} & {\small 0.0492%
} & {\small 0.0497} & {\small 0.0499} & {\small 0.0274} & {\small 0.0310} & 
{\small 0.0495} & {\small 0.0493} & {\small 0.0499} \\ 
{\small s.t.d} & {\small (0.003)} & {\small (0.004)} & {\small (0.006)} & 
{\small (0.006)} & {\small (0.003)} & {\small (0.002)} & {\small (0.003)} & 
{\small (0.005)} & {\small (0.004)} & {\small (0.003)} \\ 
${\small \beta }_{1}{\small =1}$ & {\small 1.1001} & {\small 1.0836} & 
{\small 1.0029} & {\small 0.9971} & {\small 1.0019} & {\small 1.1021} & 
{\small 1.0897} & {\small 1.0010} & {\small 1.0059} & {\small 0.9988} \\ 
{\small s.t.d} & {\small (0.068)} & {\small (0.064)} & {\small (0.067)} & 
{\small (0.060)} & {\small (0.043)} & {\small (0.047)} & {\small (0.047)} & 
{\small (0.047)} & {\small (0.046)} & {\small (0.060)} \\ 
${\small \beta }_{2}{\small =2}$ & {\small 2.0036} & {\small 2.0032} & 
{\small 2.0021} & {\small 2.0008} & {\small 1.9991} & {\small 2.0017} & 
{\small 2.0013} & {\small 1.9990} & {\small 1.9983} & {\small 2.0010} \\ 
{\small s.t.d} & {\small (0.032)} & {\small (0.031)} & {\small (0.035)} & 
{\small (0.032)} & {\small (0.020)} & {\small (0.021)} & {\small (0.020)} & 
{\small (0.022)} & {\small (0.021)} & {\small (0.030)} \\ \hline \hline
\end{tabular}
}

{\centering
\textbf{Table 1(c): Peer Effects Estimation: Large Misclassification}

 \bigskip
\begin{tabular}{c|cc|cc|c|cc|cc|c}
\hline \hline
& \multicolumn{5}{|c|}{${\small S=50}$} & \multicolumn{5}{|c}{${\small S=100}
$} \\ \hline
& \multicolumn{2}{|c|}{\small Naive} & \multicolumn{2}{|c|}{\small Adjusted}
& {\small Oracle} & \multicolumn{2}{|c}{\small Naive} & \multicolumn{2}{|c|}%
{\small Adjusted} & {\small Oracle} \\ 
{\small Reg.} & ${\small H}^{(1)}{\small y}$ & ${\small H}^{(2)}{\small y}$
& $\mathcal{H}^{(1)}{\small y}$ & $\mathcal{H}^{(2)}{\small y}$ & ${\small Gy%
}$ & ${\small H}^{(1)}{\small y}$ & ${\small H}^{(2)}{\small y}$ & $\mathcal{%
H}^{(1)}{\small y}$ & $\mathcal{H}^{(2)}{\small y}$ & ${\small Gy}$ \\ 
{\small IV} & ${\small H}^{(1)}{\small X}$ & ${\small H}^{(2)}{\small X}$ & $%
{\small H}^{(2)}{\small X}$ & ${\small H}^{(1)}{\small X}$ & ${\small GX}$ & 
${\small H}^{(1)}{\small X}$ & ${\small H}^{(2)}{\small X}$ & ${\small H}%
^{(2)}{\small X}$ & ${\small H}^{(1)}{\small X}$ & ${\small GX}$ \\ \hline
${\small n=25}$ & \multicolumn{10}{c}{\small Expected \# of peers 3.75} \\ 
\hline
${\small \lambda =0.05}$ & {\small 0.0118} & {\small 0.0180} & {\small 0.0460%
} & {\small 0.0437} & {\small 0.0489} & {\small 0.0136} & {\small 0.0195} & 
{\small 0.0532} & {\small 0.0500} & {\small 0.0508} \\ 
{\small s.t.d} & {\small (0.007)} & {\small (0.007)} & {\small (0.020)} & 
{\small (0.027)} & {\small (0.007)} & {\small (0.005)} & {\small (0.004)} & 
{\small (0.019)} & {\small (0.020)} & {\small (0.005)} \\ 
${\small \beta }_{1}{\small =1}$ & {\small 1.0813} & {\small 1.0733} & 
{\small 1.0117} & {\small 1.0173} & {\small 1.0112} & {\small 1.0822} & 
{\small 1.0722} & {\small 1.0005} & {\small 1.0189} & {\small 1.0108} \\ 
{\small s.t.d} & {\small (0.081)} & {\small (0.081)} & {\small (0.101)} & 
{\small (0.095)} & {\small (0.078)} & {\small (0.068)} & {\small (0.068)} & 
{\small (0.085)} & {\small (0.078)} & {\small (0.062)} \\ 
${\small \beta }_{2}{\small =2}$ & {\small 1.9967} & {\small 1.9980} & 
{\small 1.9951} & {\small 1.9937} & {\small 1.9946} & {\small 2.0045} & 
{\small 2.0059} & {\small 2.0023} & {\small 2.0027} & {\small 2.0018} \\ 
{\small s.t.d} & {\small (0.047)} & {\small (0.046)} & {\small (0.054)} & 
{\small (0.054)} & {\small (0.046)} & {\small (0.033)} & {\small (0.032)} & 
{\small (0.042)} & {\small (0.035)} & {\small (0.031)} \\ \hline
${\small n=50}$ & \multicolumn{10}{c}{\small Expected \# of peers 7.5} \\ 
\hline
${\small \lambda =0.05}$ & {\small 0.0132} & {\small 0.0188} & {\small 0.0510%
} & {\small 0.0510} & {\small 0.0499} & {\small 0.0133} & {\small 0.0184} & 
{\small 0.0491} & {\small 0.0486} & {\small 0.0499} \\ 
{\small s.t.d} & {\small (0.003)} & {\small (0.003)} & {\small (0.014)} & 
{\small (0.020)} & {\small (0.003)} & {\small (0.002)} & {\small (0.002)} & 
{\small (0.009)} & {\small (0.011)} & {\small (0.003)} \\ 
${\small \beta }_{1}{\small =1}$ & {\small 1.1431} & {\small 1.1273} & 
{\small 0.9942} & {\small 0.9865} & {\small 0.9988} & {\small 1.1458} & 
{\small 1.1348} & {\small 0.9956} & {\small 1.0111} & {\small 1.0019} \\ 
{\small s.t.d} & {\small (0.072)} & {\small (0.068)} & {\small (0.097)} & 
{\small (0.088)} & {\small (0.060)} & {\small (0.050)} & {\small (0.051)} & 
{\small (0.067)} & {\small (0.071)} & {\small (0.043)} \\ 
${\small \beta }_{2}{\small =2}$ & {\small 2.0011} & {\small 2.0027} & 
{\small 1.9987} & {\small 1.9995} & {\small 2.0010} & {\small 2.0000} & 
{\small 2.0010} & {\small 1.9967} & {\small 1.9976} & {\small 1.9991} \\ 
{\small s.t.d} & {\small (0.030)} & {\small (0.031)} & {\small (0.046)} & 
{\small (0.036)} & {\small (0.030)} & {\small (0.022)} & {\small (0.021)} & 
{\small (0.030)} & {\small (0.022)} & {\small (0.017)} \\ \hline \hline
\end{tabular}
}

\renewcommand{\arraystretch}{1}
\section{Application: Microfinance Participation in India}

\label{sec:application}

We apply our method to study the peer effects in household decisions to participate in a microfinance program. The sample was collected by \cite{banerjee2013diffusion} using survey questionnaires from the State of Karnataka, India between 2006-2007. 
\cite{banerjee2013diffusion} impute a social network structure in the sample by aggregating several network measures that were inferred from the survey responses. 
They study how the dissemination of information about a microfinance program, Bharatha Swamukti Samsthe, or \textit{BSS}, depended on the network position of the households that were the first to be informed about the program. 
\cite{banerjee2013diffusion} use a binary response model with social interactions to disentangle the effect of information diffusion from the peer effects, a.k.a. \textit{endorsement} effects. 
In contrast, we use two of the multiple measures in \cite{banerjee2013diffusion} as noisy measures for an actual network, and apply our method to estimate peer effects.

\subsection{Institutional background and data}

The sample was collected by \cite{banerjee2013diffusion} through survey from $43$ villages in the State of Karnataka, India.\footnote{%
    The data are publicly available at: http://economics.mit.edu/faculty/eduflo/social.} 
These villages are largely linguistically homogeneous but heterogeneous in terms of caste. 
The sample contains socioeconomic status and some demographic characteristics of 9,598 households. On average, there were about 223 households in each village, with a minimum of 114, a maximum of 356, and a standard deviation of 56.2.

We merge the information from a full-scale household census and an individual-level survey in \cite{banerjee2013diffusion}. 
The household census gathered demographic information on a variety of amenities, such as roofing material and quality of access to electric power. 
The individual survey was administered to a randomly selected sub-sample of villagers, which covered 46\% of all households in the census.
Individual questionnaires collected demographic information, such as age, caste and sub-caste, etc., but does not include explicit financial information. 
After merging the the household head information from the individual survey with the household information from the census, we have a sample of 4,149 households. 

Table 2(a) reports summary statistics for the dependent variable ($y=1$ if participates in the microfinance program) as well as a few continuous and binary explanatory variables. 
Summary statistics for additional categorical variables, such as religion, caste, property ownership, access to electricity, etc, are reported in Table 2(b).

\begin{center}
\renewcommand{\arraystretch}{0.7}

\textbf{Table 2(a): Summary of Dependent and Explanatory Variables}
\begin{tabular}{c|cccccc}
\hline
Variable & definition & mean & s.d. & min & max \\ \hline
$y$ & dummy for participation & 0.1894 & 0.3919 & 0 & 1 \\ 
$room$ & number of rooms & 2.4389 & 1.3686 & 0 & 19 \\ 
$bed$ & number of beds & 0.9229 & 1.3840 & 0 & 24 \\ 
$age$ & age of household head & 46.057 & 11.734 & 20 & 95 \\ 
$edu$ & education of household head & 4.8383 & 4.5255 & 0 & 15 \\ 
$lang$ & whether to speak other language & 0.6799 & 0.4666 & 0 & 1 \\ 
$male$ & whether the hh head is male & 0.9161 & 0.2772 & 0 & 1 \\ 
$leader$ & whether it has a leader & 0.1393 & 0.3463 & 0 & 1 \\ 
$shg$ & whether in any saving group & 0.0513 & 0.2207 & 0 & 1 \\ 
$sav$ & whether to have a bank account & 0.3840 & 0.4864 & 0 & 1 \\ 
$election$ & whether to have an election card & 0.9525 & 0.2127 & 0 & 
1 \\ 
$ration$ & whether to have a ration card & 0.9012 & 0.2985 & 0 & 1 \\ 
\hline
\end{tabular}
\end{center}

\begin{center}
\renewcommand{\arraystretch}{.7}
\textbf{Table 2(b): Summary of Category Variables}

\bigskip

\begin{tabular}{cccc|cccc}
\hline
Variable & value & obs. & per. & Variable & value & obs. & per. \\ 
\hline \hline
$religion$ & \multicolumn{3}{c|}{} & $latrine$ & \multicolumn{3}{c}{} \\ 
- & Hinduism & 3943 & 95.04 & - & Owned & 1195 & 28.80 \\ 
- & Islam & 198 & 4.77 & - & Common & 20 & 0.48 \\ 
- & Christianity & 7 & 0.19 & - & None & 2934 & 70.72 \\ \hline
$roof$ & \multicolumn{3}{c|}{} & $property$ & \multicolumn{3}{c}{ } \\ 
- & Thatch & 82 & 1.98 & - & Owned & 3727 & 89.83 \\ 
- & Tile & 1388 & 33.45 & - & Owned \& shared & 32 & 0.77 \\ 
- & Stone & 1172 & 28.25 & - & Rented & 390 & 9.40 \\ 
- & Sheet & 868 & 20.92 &  &  & \multicolumn{1}{c}{} & \multicolumn{1}{c}{}
\\ 
- & RCC & 475 & 11.45 &  &  &  &  \\ 
- & Other & 164 & 3.95 &  &  &  &  \\ \hline
$electricity$ & \multicolumn{3}{c|}{ } & $caste$ & 
\multicolumn{3}{c}{} \\ 
- & No power & 243 & 5.86  & - & Scheduled caste & 1139 & 27.54 \\ 
- & Private & 2662 & 64.18 & - & Scheduled tribe & 221 & 5.34 \\ 
- & Government & 1243 & 29.97 & - & OBC & 2253 & 54.47 \\ 
&  &  & & - & General & 523 & 12.65 \\ \hline
\end{tabular}

\end{center}

The individual-level survey in \cite{banerjee2013diffusion} also collected information about social interactions between households, including (i) individuals whose homes the respondent visited, and (ii) individuals who visited the respondent's home. 
\cite{banerjee2013diffusion} construct networks with undirected links by symmetrizing the data.\footnote{
    Two households $i$ and $j$ are considered connected by an undirected link if an individual from either household mentioned the name of someone from the other household in response to question (i). 
    Likewise, a second symmetric network measure is constructed based on responses to (ii).} 
In other words, the sample in \cite{banerjee2013diffusion} contains two symmetrized measures for the same latent network, based on the responses to (i) and (ii) respectively. 
These two measures, reported as \textquotedblleft visitGo\textquotedblright \ and \textquotedblleft visitCome\textquotedblright \ matrices in the sample and denoted as $H^{(1)}$ and $H^{(2)}$ in our notation, lend themselves to application of our method in Section \ref{sec:use_2H}.\footnote{
    \cite{banerjee2013diffusion} aggregate responses from 12 questions, including (i) and (ii), to construct a single symmetric network as the actual adjacency matrix $G$. 
    In comparison, we interpret the households' responses to (i) and (ii) as two proxies of the actual links that characterize information sharing in the structural form.}

Table 3 reports the degrees of $H^{(1)}$ and $H^{(2)}$. Because these measures are symmetric, there is no distinction between the degrees of in-bound or out-bound links. Each column lists the number of households in $H^{(1)}$ and in $H^{(2)}$ that report the number of links given by the degree column heading. 
If there were no misclassification of actual undirected links in these measures, we would expect $H^{(1)}$ and $H^{(2)}$ to be identical, and therefore have the same degree distribution. 
  Table 3 shows large differences between the two matrices in the number of reported connections between households. The fact that they differ substantially is indicative of substantial link misclassification in the measures, possibly due to the respondents' recall errors, or differences in how they interpreted the questions regarding visits.



\renewcommand{\arraystretch}{.8} 

\begin{center}

\textbf{Table 3: Degree Distribution in Two Network Measures}

\begin{tabular}{l|lllllllllll}
\hline \hline
Degree & $0$ & $1$ & $2$ & $3$ & $4$ & $5$ & $6$ & $7$ & $8$ & $9$ & $10$ \\ 
\hline
$H^{(1)}$ visit-go & 2 & 21 & 110 & 227 & 357 & 505 & 526 & 546 & 506 & 379 & 269 \\ 
$H^{(2)}$ visit-come & 4 & 24 & 112 & 245 & 384 & 522 & 534 & 577 & 491 & 386 & 255 \\ 
\hline
Degree & $11$ & $12$ & $13$ & $14$ & $15$ & $16$ & $17$ & $18$ & $19$ & $20$
& $\geq 21$ \\ \hline
$H^{(1)}$ visit-go& 224 & 145 & 90 & 74 & 54 & 33 & 27 & 15 & 9 & 6 & 24 \\ 
$H^{(2)}$ visit-come& 179 & 137 & 102 & 59 & 46 & 28 & 22 & 13 & 9 & 3 & 17 \\ 
\hline \hline
\end{tabular}

\end{center}



\subsection{Empirical strategy for estimating peer effects}

We use the following specification for the adjusted feasible structural form: 
\begin{equation}
y=\lambda W^{(t)}y+X\beta +villageFE+v^{(t)}\text{ for $t=1,2$,}
\label{eq:ESF}
\end{equation}%
where $y$ is a binary variable indicating whether the household participated in the microfinance program, $X$ is a matrix of household characteristics as listed in Table 2, and $villageFE$ are village fixed effects. Note that (\ref{eq:ESF}) provides \textit{two} different feasible structural forms (of the same actual structural model), corresponding to $t=1,2$ respectively.

Define $\phi_{ij}=1$ if households $i$ and $j$ have the same caste, and $0$ otherwise. 
We choose caste as the conditioning variable in $\phi_{ij}(\cdot)$, because it is a key demographic factor influencing the information sharing between households.
Then, based on two matrices $H^{(1)}$ (visit-go) and $H^{(2)}$ (visit-come), we get the following estimates:
\begin{eqnarray*}
& \widehat{\pi }_1=E(G_{ij}|\phi_{ij}=1)=0.0357, \text{ }
\widehat{\pi }_0=E(G_{ij}|\phi_{ij}=0)=0.0144, \\
& \widehat{p}_{0}^{(1)} =\Pr \{H_{ij}^{(1)}=1|G_{ij}=0\}=0.0020, \text{ }
\widehat{p}_{1}^{(1)} =\Pr \{H_{ij}^{(1)}=0|G_{ij}=1\}=0.1425, \\
& \widehat{p}_{0}^{(2)} =\Pr \{H_{ij}^{(2)}=1|G_{ij}=0\}=0.0001, \text{ }
\widehat{p}_{1}^{(2)} =\Pr \{H_{ij}^{(2)}=0|G_{ij}=1\}=0.1079.
\end{eqnarray*}%

Let $n_s$ be the group size of village $s$. We then construct the adjusted measures for $s=1,...,S$ and $t=1,2$: $W_{s}^{(t)}=\frac{H_{s}^{(t)}-\widehat{p}%
_{0}^{(t)}(\iota_{n_s}\iota'_{n_s}-I_{n_s})}{1-\widehat{p}_{0}^{(t)}-%
\widehat{p}_{1}^{(t)}}$, and apply our adjusted 2SLS estimator.
The estimation results are reported in Table 4, whose columns are defined as follows:

\begin{itemize}[leftmargin=*]    
    \item OLS: regression of a simple, linear model that ignores network effects by setting $\lambda =0$.
    \item (a): naive 2SLS that uses $H^{(1)}y$ as an endogenous regressor $H^{(1)}X$ as its instruments. 
    \item (b): adjusted 2SLS uses $H^{(2)}X$ as instruments for the \textit{adjusted} endogenous regressor $W^{(1)}y$. 
    \item (c): naive 2SLS analogous to (a), only with $H^{(1)}$ replaced by $H^{(2)}$.
    \item (d): adjusted 2SLS analogous to (b), uses $H^{(1)}X$ as instruments for $W^{(2)}y$.
    \item (e): S2SLS as defined in (\ref{eq:S2SLS}). This is a ``combined'' estimator that stacks the moments and associated IVs from both structural forms in (b) and (d).
\end{itemize}

In summary, columns (a) and (c) report estimators that a researcher would use if he or she ignored the issue of link misclassification, and treated either $H^{(1)} $ or $H^{(2)}$, respectively, as if it were the true adjacency matrix $G$, applying a standard 2SLS estimator in the literature. 
In contrast, columns (b), (d) and (e) report the adjusted 2SLS estimators we propose to remove the estimation bias due to link misclassification.\footnote{
    We need two network measures because the measures in the sample are symmetric. 
    As noted in Section \ref{subsec:use_1H}, we can also apply the adjusted 2SLS when the sample reports a single \textit{asymmetric} network measure.} 
Column (e) combines the information used for the estimators in (b) and (d), and so is our preferred estimator.



\subsection{Empirical results}

\begin{center}
\renewcommand{\arraystretch}{.7} 
\textbf{Table 4: Adjusted Two-stage Least Square Estimates} 


\bigskip

\begin{tabular}{c|ccc|cc|c}
\hline
& OLS & (a) & (b) & (c) & (d) & (e) \\ \hline
{\small R.h.s. Endogeneity} &  & ${\small H}^{(1)}{\small y}$ & $W%
^{(1)}{\small y}$ & ${\small H}^{(2)}{\small y}$ & $W^{(2)}{\small %
y}$ & $W^{(t)}{\small y}$ \\ 
{\small Instruments} &  & ${\small H}^{(1)}{\small X}$ & ${\small H}^{(2)}%
{\small X}$ & ${\small H}^{(2)}{\small X}$ & ${\small H}^{(1)}{\small X}$ & 
{\small Combined} \\ \hline \hline
$\widehat{\lambda }$ &  & {\small 0.0523***} & {\small 0.0499***} & {\small %
0.0550***} & {\small 0.0542***} & {\small 0.0515***} \\ 
&  & {\small (0.0079)} & {\small (0.0086)} & {\small (0.0097)} & {\small %
(0.0082)} & {\small (0.0083)} \\ \hline
${\small leader}$ & {\small 0.0515***} & {\small 0.0371**} & {\small 0.0355**%
} & {\small 0.0414**} & {\small 0.0403**} & {\small 0.0379**} \\ 
& {\small (0.0175)} & {\small (0.0187)} & {\small (0.0188)} & {\small %
(0.0184)} & {\small (0.0184)} & {\small (0.0185)} \\ \hline
${\small age}$ & {\small -0.0012***} & {\small -0.0017***} & {\small %
-0.0017***} & {\small -0.0016***} & {\small -0.0017***} & {\small -0.0017***}
\\ 
& {\small (0.0005)} & {\small (0.0005)} & {\small (0.0005)} & {\small %
(0.0005)} & {\small (0.0005)} & {\small (0.0005)} \\ \hline
${\small ration}$ & {\small 0.0502**} & {\small 0.0438**} & {\small 0.0430**}
& {\small 0.0420**} & {\small 0.0412**} & {\small 0.0422**} \\ 
& {\small (0.0212)} & {\small (0.0201)} & {\small (0.0202)} & {\small %
(0.0195)} & {\small (0.0194)} & {\small (0.0198)} \\ \hline
${\small electricity-gov}$ & {\small 0.0441**} & {\small 0.0338**} & {\small %
0.0326**} & {\small 0.0349**} & {\small 0.0339**} & {\small 0.0333**} \\ 
& {\small (0.0152)} & {\small (0.0157)} & {\small (0.0158)} & {\small %
(0.0156)} & {\small (0.0155)} & {\small (0.0156)} \\ \hline
${\small electricity-no}$ & {\small 0.0162} & {\small 0.0226} & {\small %
0.0233} & {\small 0.0240} & {\small 0.0248} & {\small 0.0240} \\ 
& {\small (0.0275)} & {\small (0.0296)} & {\small (0.0296)} & {\small %
(0.0300)} & {\small (0.0298)} & {\small (0.0297)} \\ \hline
${\small caste-tribe}$ & {\small -0.0411} & {\small -0.0278} & -{\small %
0.0263} & {\small -0.0270} & -{\small 0.0255} & -{\small 0.0260} \\ 
& {\small (0.0294)} & {\small (0.0309)} & {\small (0.0305)} & {\small %
(0.0301)} & {\small (0.0298)} & {\small (0.0301)} \\ \hline
${\small caste-obc}$ & -{\small 0.0822***} & -{\small 0.0505**} & -{\small %
0.0468**} & -{\small 0.0472**} & -{\small 0.0435***} & -{\small 0.0456***}
\\ 
& {\small (0.0163)} & {\small (0.0217)} & {\small (0.0214)} & {\small %
(0.0218)} & {\small (0.0210)} & {\small (0.0212)} \\ \hline
${\small caste-gen}$ & {\small -0.1142***} & {\small -0.0718***} & {\small %
-0.0669***} & {\small -0.0669***} & {\small -0.0620**} & {\small -0.0650***}
\\ 
& {\small (0.0239)} & {\small (0.0238)} & {\small (0.0244)} & {\small %
(0.0244)} & {\small (0.0235)} & {\small (0.0241)} \\ \hline
${\small religion-Islam}$ & {\small 0.1225***} & {\small 0.0967***} & 
{\small 0.0938***} & {\small 0.0880***} & {\small 0.0843***} & {\small %
0.0895***} \\ 
& {\small (0.0332)} & {\small (0.0325)} & {\small (0.0325)} & {\small %
(0.0346)} & {\small (0.0349)} & {\small (0.0335)} \\ \hline
${\small religion-Chri}$ & {\small 0.1569} & {\small 0.1427} & {\small 0.1410%
} & {\small 0.1462} & {\small 0.1450} & {\small 0.1431} \\ 
& {\small (0.1440)} & {\small (0.1295)} & {\small (0.1279)} & {\small %
(0.1310)} & {\small (0.1299)} & {\small (0.1287)} \\ \hline
${\small Controls}$ & ${\small \surd }$ & ${\small \surd }$ & ${\small \surd 
}$ & ${\small \surd }$ & ${\small \surd }$ & ${\small \surd }$ \\ 
${\small VillageFE}$ & ${\small \surd }$ & ${\small \surd }$ & ${\small %
\surd }$ & ${\small \surd }$ & ${\small \surd }$ & ${\small \surd }$ \\ 
${\small R}^{2}$ & {\small 0.0862} & {\small 0.1339} & {\small 0.1353} & 
{\small 0.1356} & {\small 0.1366} & {\small 0.1358} \\ 
{\small Obs} & {\small 4134} & {\small 4134} & {\small 4134} & {\small 4134}
& {\small 4134} & {\small 4134} \\ \hline
\multicolumn{7}{c}{\small Note: s.e. clustered at village level are in parentheses. ***, **, and * indicate
1\%, 5\% and 10\% significant.} \\ 
\multicolumn{7}{c}{{\small Controls include }${\small male}${\small , }$%
{\small roof}${\small , }${\small room}${\small , }${\small bed}${\small , }$%
{\small latrine}${\small , }${\small edu}${\small , }${\small lang}${\small %
, }${\small shg}${\small , }${\small sav}${\small , }${\small election}$%
{\small , }${\small own}${\small .}}%
\end{tabular}
\end{center}

Table 4 reports that our adjusted 2SLS estimates for the peer effect $\widehat{\lambda }$ are 0.0499 when using $W^{(1)}y$ in the structural form (column (b)), 0.0542 using $W^{(2)}y$ (column (d)), and 0.0515 using both measures and S2SLS (column (e)), all significant at the 1\% level, and the differences between them are small relative to their standard errors. 
These estimates imply the likelihood of a household to participate in the microfinance program is increased by about 5.15\% when the household is linked to one more participating household on the network. 
With the average participation rate being 18.9\% in the sample, these estimates suggest that peer effects
are substantial.

The signs of estimated marginal effects by individual or household characteristics are plausible. 
Column (e) suggests the head of household being a \textquotedblleft leader\textquotedblright \ (e.g. a teacher, a leader of a self-help group, or a shopkeeper) increases the participation rate by around 3.8\%. 
These households with \textquotedblleft leaders\textquotedblright \ were the first ones to be informed about the program, and were asked to forward information about the microfinance program to other potentially interested villagers. 
Leaders had received first-hand, detailed information about the program from its administrator, which could be conducive to higher participation rates. Households with younger heads are more likely to participate, 
but the magnitude of this age effect is less substantial. 
Being 10 years younger increases the participation rate by 1.7\%. Having a ration card increases the participation rate by around 4.2\%. 
Compared to households using private electricity, households using government-supplied electricity have a 3.3\% higher participation rate. 
The two factors indicate that holding other things equal, households in poorer economic conditions are
more inclined to participate in the microfinance program. 

Table 4 also shows that, if we had ignored the issue of misclassified links
in network measures, and had done 2SLS using $H^{(t)}X$ as instruments for
the (un-adjusted) endogenous peer outcomes $H^{(t)}y$, then the estimator
would have been biased. In (a), where we use $H^{(1)}X$ as instruments for $%
H^{(1)}y$, the estimate for $\lambda $ is 0.0523. In comparison, in (b)
where we correct for misclassified link bias by using $H^{(2)}X$ as
instruments for $W^{(1)}y$, then the estimated $\lambda $ is
0.0499. The upward bias resulted from ignoring the misclassified links is
about 4.8\% (as 0.0523/0.0499=1.048). Likewise, in (c) where we erroneously
use $H^{(2)}X$ as instruments for $H^{(2)}y$, we get an upward bias about
1.5\% in the peer effect estimate compared with the correct estimate in (d)
(as 0.0550/0.0542=1.015).

As explained in Section \ref{sec:augmentBias}, the bias in (a) and (c) is due to the correlation between $H^{(t)}X$ and the composite errors $\varepsilon + \lambda [G-H^{(t)}] y$. 
The magnitude of this bias is determined in part by the misclassification rates $(p_{0}^{(t)},p_{1}^{(t)})$, which affect the correlation between the composite errors and the traditional instruments $H^{(t)}X$ for endogenous peer outcomes $H^{(t)}y$ in a naive 2SLS. 
This is evident from (\ref{eq:relate_compErr}): if both $p_0$ and $p_1$ were close to zero, then the right hand side of (\ref{eq:relate_compErr}) would be almost reduced to $v$, which is mean independent from $X$ under Lemma \ref{lm:exogenous_v}. 
In that case, $H^{(3-t)}X$ would be valid IVs for $H^{(t)}y$ even without making adjustments in $W^{(t)}$.

The fact that estimates in (a) and (c) are fairly close to those in (b), (d) and (e) indicate the impact of link misclassification on peer effects is relatively low in this application.
However, our Monte Carlo simulations sometimes showed much larger impacts from misclassification, which suggests that in other empirical environments, we may expect larger bias when misclassification of links is not accounted for in estimation. The method we propose in this paper provides an easy remedy for this issue.

Table 4 suggests significant, positive endogenous peer effects around 5\% across various specification, while \cite{banerjee2013diffusion} find no significant ``endorsement effects'' \textit{after} controlling for information passing between the households.\footnote{\cite{banerjee2013diffusion} define ``endorsement effects'' as the impact of friends' decisions (to adopt a product) on the decisions of informed individuals \textit{within} a diffusion process.}
This difference arises because, in contrast with \cite{banerjee2013diffusion}, we do not separately account for an additional layer of structure that gives rise to information diffusion.
In this sense, our model is more ``reduced-form'' than that of \cite{banerjee2013diffusion}. 
Therefore, our estimates for $\lambda$ could be interpreted as a compound of what they define as the endorsement effect and the effect of information diffusion. 
The latter is indeed found to be statistically significant by \cite{banerjee2013diffusion}.

We conclude this section with model validation results in Table 5, which shows how the predicted values of $E(y|X)$ fit with the sample data. 
The Probit and Logit models use the same set of regressors as in Table 4. 
We report the summary statistics of the fitted values $\widehat{E(y|X)}$ under different models. 
Columns (a) through (d) of Table 5 are the fitted values of the feasible structural models used in each of the corresponding columns in Table 4.

In all but one of the models in Table 5, the sample mean of the predicted participation probability $\widehat{E(y|X)}$ is 0.1894, which is equal to the sample mean of $y$ in the 4,134 observations used in the regression. 
The standard deviation of the predicted participation probability varies across different models. 
Predictions of linear probability models (LPM), reported under the column of \textquotedblleft OLS\textquotedblright \ and (a)-(e), are mostly within the unit interval $[0,1]$. LPM predictions are strictly less than $1$ for all observations in the sample; 
only 2.95\% to 5.56\% of the households in the sample end up with negative LPM predictions. That is, about 95\% all LPM predictions in the sample are indeed within the unit interval.

Based on $\widehat{E(y|X)}$, we use the indicator $\mathds{1}\{\widehat{E(y|X)}>0.5\}$ to
predict whether an individual participates in the microfinance program, and
calculate prediction rates. Predictions in our linear social network models
in columns (a)-(e) generally outperform the OLS, Probit and Logit models in
terms of the percentage of correct predictions. 
\begin{center}
\renewcommand{\arraystretch}{.7} 
\textbf{Table 5: Model Validation: Predicted Microfinance Participation}

\bigskip 

\begin{tabular}{c|c|c|c|cc|cc|c}
\hline \hline
$\widehat{E(y|X)}$ & {\small Probit} & {\small Logit} & {\small OLS} & 
\multicolumn{2}{|c|}{(a)\space \space \space \space \space \space(b)} & 
\multicolumn{2}{|c|}{(c)\space \space \space \space \space \space(d)} & {\small %
(e)} \\ \hline
${\small mean}$ & {\small 0.1894} & {\small 0.1894} & {\small 0.1894} & 
{\small 0.1894} & {\small 0.1894} & {\small 0.1894} & {\small 0.1894} & 
{\small 0.1894} \\ \hline
${\small s.t.d}$ & {\small 0.1176} & {\small 0.1181} & {\small 0.1151} & 
{\small 0.1357} & {\small 0.1403} & {\small 0.1372} & {\small 0.1416} & 
{\small 0.1405} \\ \hline
$\min $ & {\small 0.0103} & {\small 0.0166} & {\small -0.0953} & {\small %
-0.1062} & {\small -0.1107} & {\small -0.1282} & {\small -0.1316} & {\small %
-0.1314} \\ \hline
$\max $ & {\small 0.7490} & {\small 0.7673} & {\small 0.6895} & {\small %
0.7911} & {\small 0.8159} & {\small 0.7370} & {\small 0.7615} & {\small %
0.8286} \\ \hline
${\small <0}$ & {\small 0\%} & {\small 0\%} & {\small 2.95\%} & {\small %
4.96\%} & {\small 5.32\%} & {\small 5.06\%} & {\small 5.56\%} & {\small %
5.41\%} \\ \hline \hline
${\small I\{}\widehat{E(y|X)}{\small >0.5\}}$ &  &  &  &  &  &  &  &  \\ 
\hline
{\small underpred.}${\small \text{ (1 to 0)}}$ & {\small 17.76\%} & {\small %
17.66\%} & {\small 18.34\%} & {\small 17.27\%} & {\small 17.05\%} & {\small %
17.30\%} & {\small 17.08\%} & {\small 17.10\%} \\ \hline
{\small overpred.}${\small \text{ (0 to 1)}}$ & {\small 0.92\%} & {\small %
1.11\%} & {\small 0.27\%} & {\small 0.94\%} & {\small 1.14\%} & {\small %
0.87\%} & {\small 1.92\%} & {\small 1.04\%} \\ \hline
{\small correct} & {\small 81.33\%} & {\small 81.23\%} & {\small 81.40\%} & 
{\small 81.79\%} & {\small 81.81\%} & {\small 81.83\%} & {\small 81.91\%} & 
{\small 81.86\%} \\ \hline
\end{tabular}

\end{center}

\section{Conclusion}

This paper proposes adjusted-2SLS estimators that consistently estimate structural parameters, including peer effects, in social networks when the reported links are subject to random misclassification errors. 
By adjusting the endogenous peer outcomes and applying new instruments constructed from noisy network measures, our estimators resolve the additional endogeneity issues caused by link misclassification. 
As an initial step of our method, we propose simple, closed-form estimators for the misclassification rates in the network measures. 

We apply our method to analyze the peer effects in households' decisions to participate in a microfinance program in Indian villages, using the data collected by \cite{banerjee2013diffusion}. Consistent with our theory, our empirical estimates show that ignoring the issue of misclassified links in 2SLS estimation of social network models leads to an upward bias of up to 5\% in the estimates of peer effects. A Monte Carlo analysis shows that in other applications, the bias from failing to account for link misclassification can be much larger.
\begin{spacing}{1}

\addcontentsline{toc}{chapter}{References}
\bibliographystyle{chicago}
\bibliography{missing_links}
\end{spacing}

\section*{Appendix}\label{sec: appendix}

\subsection*{A1. Proofs in Sections \ref{sec:assum_H}-\ref{sec:use_H}}

\begin{proof}[Proof of Lemma \protect\ref{lm:exogenous_v}]
Under (A3), $E(Gy|G,X)=E[GM(X\beta + \varepsilon )|G,X]=GMX\beta $, and $E(Wy|G,X)=E[WME\left( X\beta + \varepsilon |H,G,X\right) |G,X]=E(W|G,X)MX\beta$. 
Under (A1) and (A2), $E(W|G,X)=G$. 
It follows from the definition of $v$ in (\ref{SF1}) that $E(v|G,X)=0$.\bigskip
\end{proof}

\begin{proof}[Proof of Proposition \protect\ref{pn:HX}]
By (A1), (A2), (A4), conditional mean of $(i,j)$-th entry in $%
W^{2}$ is 
\begin{eqnarray}
&&E\left[ (W^{2})_{ij}|G,X\right] =E\left( \left.
\sum\nolimits_{k\neq i,j}W_{ik}W_{kj}\right\vert
G,X\right) =\sum\nolimits_{k\neq i,j}E\left( \left. W_{ik}{%
W}_{kj}\right\vert G,X\right)  \nonumber \\
&=&\sum\nolimits_{k\neq i,j}E\left( \left. W_{ik}\right\vert
G_{ik},X\right) E\left( \left. W_{kj}\right\vert G_{kj},X\right) = \sum\nolimits_{k\neq i,j}G_{ik}G_{kj} = \left( G^{2}\right) _{ij}\text{.}  \label{eq:H2}
\end{eqnarray}%
It then follows that 
\begin{eqnarray}
&&E[(W^{\prime }X)^{\prime }v|G,X] = E\left[ \left. X^{\prime }W(\varepsilon+\lambda\left( G-W\right) y)\right\vert G,X\right]=\lambda E\left[ \left. X^{\prime }W\left( G-W\right)
MX\beta \right\vert G,X\right]  \nonumber \\
&=&\lambda X^{\prime }\left[ E(W|G,X)G-E(W^{2}|G,X)%
\right] MX\beta  =\lambda X^{\prime }\left( G^{2}-G^{2}\right) MX\beta =0\text{,}
\label{eq:HH}
\end{eqnarray}%
where the first two equalities are due to (A3) and the reduced form of $y$, and the last due to (\ref%
{eq:H2}) and the fact that $E\left[ W|G,X\right] =G$ under (A1)
and (A2).

As $H=(1-p_{0}-p_{1})W+p_{0}(\iota \iota'
-I) $, $E[(H^{\prime }X)^{\prime }v|G,X]=0+E\left\{ \left. X^{\prime }p_{0}(\iota
\iota' -I)v\right\vert G,X\right\} =0 $, where the first equality is due to (\ref{eq:HH}) and the second due to Lemma \ref{lm:exogenous_v}.\bigskip
\end{proof}

As noted in Section \ref{sec:use_2H}, we can construct instruments from
multiple \textit{symmetrized} measures for $G$, denoted by $H^{(1)}$ and $%
H^{(2)}$. 
Suppose $H^{(1)}$ and $H^{(2)}$ both satisfy (A1), (A2), (A3), and are
independent in the sense of (A4'). Let $W^{(t)}$ be defined for $%
t=1,2$ as in the text.

We can construct feasible structural forms as in (\ref{eq:2HSF}), and use $%
W^{(2)}X$ (or $H^{(2)}X$) as instruments for $v^{(1)}$. To see
why, note that for all $i$ and $j$ (including the case with $i=j$):%
\begin{eqnarray}
&&E\left[ (W^{(2)}W^{(1)})_{ij}|G,X\right] =E\left(
\left. \sum\nolimits_{k\neq i,j}W_{ik}^{(2)}W%
_{kj}^{(1)}\right\vert G,X\right)=\sum\nolimits_{k\neq i,j}E\left( \left. W_{ik}^{(2)}W%
_{kj}^{(1)}\right\vert G,X\right) \nonumber\\&& =\sum\nolimits_{k\neq i,j}E\left( \left. 
W_{ik}^{(2)}\right\vert G_{ik},X\right) E\left( \left. W%
_{kj}^{(1)}\right\vert G_{kj},X\right) =\sum\nolimits_{k\neq i,j}G_{ik}G_{kj}=\left( G^{2}\right) _{ij}\text{.}
\label{eq:HA2}
\end{eqnarray}%
Under (A1)-(A2), $E\left( W^{(2)}G|G,X\right) =E(W^{(2)}|G,X)G=G^{2}$.
It then follows that
\begin{eqnarray*}
&&E[(W^{(2)}X)^{\prime }v^{(1)}|G,X] =E(X^{\prime }W%
^{(2)}\varepsilon |G,X)+\lambda E\left\{ \left. X^{\prime }W^{(2)}%
\left[ G-W^{(1)}\right] y\right\vert G,X\right\} \\
&=&\lambda E\left[ \left. X^{\prime }W^{(2)}\left( G-W%
^{(1)}\right) MX\beta \right\vert G,X\right]=\lambda X^{\prime }E\left[ (W^{(2)}G-W^{(2)}%
W^{(1)})|G,X\right] MX\beta =0\text{.}
\end{eqnarray*}%

\begin{proof}[Proof of Proposition \protect\ref{pn:IVR}]
Define some $K$-by-$K$ moments involving $(G,X)$: 
\begin{eqnarray*}
B_{1} \equiv E(X^{\prime }G^{2}MX)\text{, }B_{2}\equiv E(X^{\prime }GMX)%
\text{, }B_{3}\equiv E(X^{\prime }G^{2}X)\text{, } B_{4} &\equiv &E(X^{\prime }GX)\text{, }B_{5}\equiv E(X^{\prime }X)\text{.}
\end{eqnarray*}

Recall $Z\equiv (W^{\prime }X,X)$ and $R\equiv (Wy,X)$.
Under (A1), (A2), (A3), and (A4), 
\begin{eqnarray*}
E(Z^{\prime }R) &=&\left( 
\begin{array}{cc}
E(X^{\prime }W^{2}y) & E(X^{\prime }WX) \\ 
E(X^{\prime }Wy) & E(X^{\prime }X)%
\end{array}%
\right) =\left( 
\begin{array}{cc}
E(X^{\prime }G^{2}MX\beta ) & E(X^{\prime }GX) \\ 
E(X^{\prime }GMX\beta ) & E(X^{\prime }X)%
\end{array}%
\right) \equiv \left( 
\begin{array}{cc}
B_{1}\beta & B_{4} \\ 
B_{2}\beta & B_{5}%
\end{array}%
\right) \text{.}
\end{eqnarray*}%
If $E(Z^{\prime }R)$ does not have full
rank, it implies the singularity of the $2K$-by-$2K$ square matrix 
\begin{equation}
\left( 
\begin{array}{cc}
B_{1} & B_{4} \\ 
B_{2} & B_{5}%
\end{array}%
\right)  \label{SQM}
\text{.}\end{equation}%
Hence, the non-singularity of the square matrix
in (\ref{SQM}) implies $E(Z^{\prime }R)$ has full rank.

As $M-\lambda GM=I$, $GM=\lambda ^{-1}(M-I)$ and $G^{2}M=\lambda
^{-1}(GM-G)=\lambda ^{-2}(M-I-\lambda G)$,
\[
\left( 
\begin{array}{cc}
B_{1} & B_{4} \\ 
B_{2} & B_{5}%
\end{array}%
\right) =\left( 
\begin{array}{cc}
\lambda ^{-1}E[X^{\prime }(GM-G)X] & E(X^{\prime }GX) \\ 
E(X^{\prime }GMX) & E(X^{\prime }X)%
\end{array}%
\right) . 
\]%
Adding (the 2nd row)$\times(-\frac{1}{\lambda })$ to the 1st
row, we get $\left( 
\begin{array}{cc}
-\frac{1}{\lambda }E(X^{\prime }GX) & E(X^{\prime }GX)-\frac{1}{\lambda }%
E(X^{\prime }X) \\ 
E(X^{\prime }GMX) & E(X^{\prime }X)%
\end{array}%
\right)$. 

Adding (the 2nd column) $\times\frac{1}{\lambda }$ to the 1st
column, we get 
\[
\left( 
\begin{array}{cc}
-\frac{1}{\lambda ^{2}}E(X^{\prime }X) & E(X^{\prime }GX)-\frac{1}{\lambda }%
E(X^{\prime }X) \\ 
E(X^{\prime }(GM+\frac{1}{\lambda }I)X) & E(X^{\prime }X)%
\end{array}%
\right) =\left( 
\begin{array}{cc}
-\frac{1}{\lambda ^{2}}E(X^{\prime }X) & -\frac{1}{\lambda }E(X^{\prime
}M^{-1}X) \\ 
\frac{1}{\lambda }E(X^{\prime }MX) & E(X^{\prime }X)%
\end{array}%
\right) . 
\]%
The determinant of the matrix on the right-hand side is the product of $%
\lambda ^{-2K}$ and the determinant of $[E(X^{\prime }X),E(X^{\prime
}M^{-1}X);E(X^{\prime }MX),E(X^{\prime }X)]$. Hence, the matrix in (\ref{SQM}%
) is non-singular iff $[E(X^{\prime }X),E(X^{\prime }M^{-1}X);E(X^{\prime
}MX),E(X^{\prime }X)]$ is non-singular.

By the same token, (A1), (A2), and (A4) imply that 
\[
E(Z^{\prime }Z)=\left( 
\begin{array}{cc}
E(X^{\prime }W^{2}X) & E(X^{\prime }WX) \\ 
E(X^{\prime }WX) & E(X^{\prime }X)%
\end{array}%
\right) =\left( 
\begin{array}{cc}
E(X^{\prime }G^{2}X) & E(X^{\prime }GX) \\ 
E(X^{\prime }GX) & E(X^{\prime }X)%
\end{array}%
\right) =\left( 
\begin{array}{cc}
B_{3} & B_{4} \\ 
B_{4} & B_{5}%
\end{array}%
\right) \text{.} 
\]%
Therefore, $E(Z^{\prime }Z)$ has full rank if and only if $%
[B_{3},B_{4};B_{4},B_{5}]$ is non-singular.\bigskip
\end{proof}

\subsection*{A2. Identifying misclassification rates in Section \protect \ref{sec:get_p}}

\label{subsec:A2}

Consider the case in Section \ref{subsubsec: 2H} where the sample reports
two measures $H^{(t)}$ with misclassification rates $%
p_{0}^{(t)},p_{1}^{(t)}$ for $t=1,2$ respectively. Introduce $H_{ij}^{(3)}\equiv \max \left \{ H_{ij}^{(1)},H_{ij}^{(2)}\right \} $.

By construction, for $t=1,2,3$, the distribution of $H_{ij}^{(t)}$ is
related to $p_{0}^{(t)},p_{1}^{(t)}$ and link formation probability $\pi
_{1}\equiv E(G_{ij}|\phi _{ij}(X)=1)$ as follows:%
\begin{equation}
\psi _{1}^{(t)}\equiv E\left[ \left. H_{ij}^{(t)}\right \vert \phi _{ij}(X)=1%
\right] =\left( 1-p_{1}^{(t)}\right) \pi _{1}+p_{0}^{(t)}(1-\pi
_{1})=p_{0}^{(t)}+\left( 1-p_{1}^{(t)}-p_{0}^{(t)}\right) \pi _{1}\text{,}
\label{eq:psi_ta}
\end{equation}%
where (A4') implies: $p_{0}^{(3)} =p_{0}^{(1)}+p_{0}^{(2)}-p_{0}^{(1)}p_{0}^{(2)}$, and $p_{1}^{(3)} =p_{1}^{(1)}p_{1}^{(2)}$. 

Then $\psi _{0}^{(t)}$ is defined by replacing \textquotedblleft $\phi
_{ij}(X)=1$\textquotedblright \ and $\pi _{1}$ by \textquotedblleft $%
\phi _{ij}(X)=0$\textquotedblright \ and $\pi _{0}$ in (\ref{eq:psi_ta}).
\bigskip
\noindent \textbf{[Identifying }$p_{0}^{(1)}$\textbf{\ and }$p_{0}^{(2)}$%
\textbf{.]} For convenience, let $\xi _{1}\equiv \left(
1-p_{0}^{(2)}-p_{1}^{(2)}\right) \pi _{1}$ so that
\begin{eqnarray}
\psi _{1}^{(1)} =p_{0}^{(1)}+r_{(12)}\xi _{1}\text{; \ }\psi
_{1}^{(2)}=p_{0}^{(2)}+\xi _{1}\text{; } 
\psi _{1}^{(3)}=p_{0}^{(1)}+p_{0}^{(2)}-p_{0}^{(1)}p_{0}^{(2)}+r_{(32)}\xi _{1}\text{,}
\label{eq:defn_psi}
\end{eqnarray}%
where $r_{(t^{\prime }t)}\equiv (\psi _{0}^{(t^{\prime })}-\psi
_{1}^{(t^{\prime })})/(\psi _{0}^{(t)}-\psi _{1}^{(t)})$ for $t^{\prime
},t\in \{1,2,3\}$. This implies
\begin{equation}
p_{0}^{(1)}=\psi _{1}^{(1)}-r_{(12)}\xi _{1}\text{ and }p_{0}^{(2)}=\psi
_{1}^{(2)}-\xi _{1}\text{.}  \label{eq:P0s}
\end{equation}%
Substituting these into $\psi _{1}^{(3)}$ in (\ref%
{eq:defn_psi}) implies: $\psi _{1}^{(3)}=\left( \psi _{1}^{(1)}-r_{(12)}\xi _{1}-1\right) \left(
1-\psi _{1}^{(2)}+\xi _{1}\right) +1+r_{(32)}\xi _{1}$.
Rearranging terms, we write this quadratic equation in $\xi _{1}$ as
\begin{equation}
\mathcal{C}_{2}\xi _{1}^{2}-\mathcal{C}_{1}\xi _{1}-\mathcal{C}_{0}=0\text{,}
\label{eq:QE}
\end{equation}%
where $\mathcal{C}_{2} \equiv r_{(12)}\text{, } \mathcal{C}_{1} \equiv \psi _{1}^{(1)}-1+r_{(32)}-r_{(12)}(1-\psi
_{1}^{(2)})\text{, } \mathcal{C}_{0} \equiv \psi _{1}^{(1)}+\psi _{1}^{(2)}-\psi _{1}^{(1)}\psi
_{1}^{(2)}-\psi _{1}^{(3)}$.

Then $\Delta \equiv (\mathcal{C}_{1})^{2}+4\mathcal{C}_{2}\mathcal{C}_{0}>0$
and $\sqrt{\Delta }>\mathcal{C}_{1}$ as $\mathcal{C}%
_{2}=[1-p_{0}^{(1)}-p_{1}^{(1)}]/[1-p_{0}^{(2)}-p_{1}^{(2)}]>0$, and $\mathcal{C}_{0}=[ 1-p_{0}^{(1)}-p_{1}^{(1)}] [
1-p_{0}^{(2)}-p_{1}^{(2)}] \pi _{1}(1-\pi _{1})>0$. Hence, (\ref{eq:QE}) has two
solutions in $\xi _{1}$: $\xi _{1}=\frac{1}{2\mathcal{C}_{2}}(\mathcal{C}_{1}\pm \sqrt{\Delta })$. Furthermore, $\frac{1}{2\mathcal{C}_{2}}%
\left( \mathcal{C}_{1}-\sqrt{\Delta }\right) <0$, so the only solution in (\ref%
{eq:QE}) is $\xi _{1}=\frac{1}{2\mathcal{C}_{2}}\left( \mathcal{C}_{1}+%
\sqrt{\Delta }\right) $. Plugging in this solution of $\xi _{1}$ into (\ref%
{eq:P0s}) identifies $p_{0}^{(1)}$ and $p_{0}^{(2)}$.\bigskip

\noindent \textbf{[Identifying }$\pi _{1}$\textbf{.]} Note that (\ref%
{eq:psi_ta}) implies
\begin{equation}
p_{1}^{(t)}=1-p_{0}^{(t)}-\frac{\psi _{1}^{(t)}-p_{0}^{(t)}}{\pi _{1}}\text{
for }t=1,2,3.  \label{eq:p1t}
\end{equation}%
Recall that $\psi _{1}^{(t)}$ for $t=1,2,3$ are directly identified from the
data. With $p_{0}^{(t)}$ identified for $t=1,2$, we can recover $p_{0}^{(3)}$
. This implies $\pi _{1}$ is identified as:%
\begin{equation}
\pi _{1}=\frac{\left( \psi _{1}^{(1)}-p_{0}^{(1)}\right) \left( \psi
_{1}^{(2)}-p_{0}^{(2)}\right) }{\left( 1-p_{0}^{(1)}\right) \left( \psi
_{1}^{(2)}-p_{0}^{(2)}\right) +\left( 1-p_{0}^{(2)}\right) \left( \psi
_{1}^{(1)}-p_{0}^{(1)}\right) -\left( \psi _{1}^{(3)}-p_{0}^{(3)}\right) }%
\text{.}  \label{eq:pi_a}
\end{equation}%

\bigskip
\noindent \textbf{[Identifying }$p_{1}^{(1)}$, $p_{1}^{(2)}$\textbf{\ and }$%
\pi _{0}$\textbf{.]} With $p_{0}^{(1)},p_{0}^{(2)}$, and $\pi _{1}$
identified above, we can use (\ref{eq:p1t})\ to recover $p_{1}^{(t)}$ from $%
\psi _{1}^{(t)}$ for $t=1,2$ and use (\ref{eq:psi_ta}) to
identify $\pi _{0}$ as: $\pi _{0}=\frac{\psi _{0}^{(t)}-p_{0}^{(t)}}{\psi _{1}^{(t)}-p_{0}^{(t)}}\pi
_{1}\text{ for }t=1,2,3\text{.}$

\bigskip

\noindent \textbf{[Single, unsymmetrized measure.]} The same argument applies for the case in Section \ref{subsubsec: 1H}, where the
sample reports a single, unsymmetrized measure $H$ with misclassification
rates $p_{1},p_{0}$ when the actual $G$ is symmetric. For each
\textit{unordered} pair $\{i,j\}$, define $H_{\{i,j\}}^{(1)}\equiv H_{ij}$, $%
H_{\{i,j\}}^{(2)}\equiv H_{ji}$, and $H_{\{i,j\}}^{(3)}\equiv \max \{H_{ij},H_{ji}\}$.
There exists a system analogous to (\ref{eq:psi_ta}), with $H_{ij}^{(t)}$
replaced by $H_{\{i,j\}}^{(t)}$. 
However, the first two equations for $t=1,2$ coincide with each other by construction. 
The remaining steps for identification are identical to the case above with two measures $H^{(1)}$ and $H^{(2)}$, except that the closed-form expressions are further simplified due to $r_{(12)}=1$, $\psi _{1}^{(1)}=\psi _{1}^{(2)}$, and $p_{d}^{(1)}=p_{d}^{(2)}$ for $d\in \{0,1\}$.

\end{spacing}

\end{document}